\documentclass[11pt]{article}

\usepackage[margin=1.05in]{geometry}
  
\usepackage[T1]{fontenc}
\usepackage{color,graphicx}
\usepackage{array}
\usepackage{enumerate}
\usepackage{amsmath}
\usepackage{amssymb}
\usepackage{amsthm}
\usepackage{pgfplots}
\usepackage{pgf}
\usepackage{tikz}
\usetikzlibrary{patterns}
\usepgfplotslibrary{patchplots} % LATEX and plain TEX
\usetikzlibrary{pgfplots.patchplots} % LATEX and plain TEX
\pgfplotsset{width=9cm,compat=1.5.1}

\usepackage[english]{babel}
\usepackage{amsthm}
\newtheorem{theorem}{Theorem}[section]

\newtheorem{lemma}[theorem]{Lemma}
\newtheorem{corollary}[theorem]{Corollary}
\newtheorem{remark}[theorem]{Remark}
\newtheorem{proposition}[theorem]{Proposition}

\newtheorem{example}[theorem]{Example} 

\usepackage{mathpazo}
\usepackage{amsfonts}
\usepackage{amsmath}

\usepackage{graphicx}
\usepackage{latexsym}

\newcommand{\zero}{\mathbf 0}
\newcommand{\one}{\mathbf 1}
\newcommand{\be}{\mathbf e}
\newcommand{\bff}{\mathbf f}
\newcommand{\bh}{\mathbf h}

\newcommand{\bu}{\mathbf u}
\newcommand{\bv}{\mathbf v}
\newcommand{\bxx}{\mathbf x}
\newcommand{\by}{\mathbf y}
\newcommand{\bw}{\mathbf w}

\newcommand{\bmu}{{\boldsymbol{\mu}}}
\newcommand{\bnu}{{\boldsymbol{\nu}}}

\newcommand{\bomega}{{\boldsymbol{\omega}}}
\newcommand{\supp}{{\Phi}}
\newcommand{\Lsupp}{{\Lambda}}
\newcommand\spn{{\mathrm{span\ }}}
\newcommand\dist{{\mathrm{dist}}}
\newcommand\spec{{\mathrm{spec}}}
\newcommand\rank{{\mathrm{rank}}}

\newcommand{\Line}[1]{{{\cal L}(#1)}}
\newcommand{\LineA}{{A_{\cal L}}}
\newcommand{\LineQ}{{Q_{\cal L}}}
\newcommand{\wt}[1]{{\widetilde{#1}}}

\newcommand{\etalchar}[1]{$^{#1}$}

\def\C{\mathbb{C}}
\def\R{\mathbb{R}}
\def\Z{\mathbb{Z}}
\def\Q{\mathbb{Q}}
\def\ii{\mathrm{i}}

\newcommand{\bra}[1]{\langle #1 \rvert}
\newcommand{\ket}[1]{\lvert #1 \rangle}

\newcommand{\ketbra}[2]{\lvert #1 \rangle\langle #2 \rvert}

\usepackage{makeidx}
\newcommand{\tr}{\operatorname{Tr}}
\newcommand\ignore[1]{{}}

\begin{document}

%%%%%%%%%%%%%%%%%%%%%%%
%% Elsevier bibliography styles
%%%%%%%%%%%%%%%%%%%%%%%
%% To change the style, put a % in front of the second line of the current style and
%% remove the % from the second line of the style you would like to use.
%%%%%%%%%%%%%%%%%%%%%%%

%% Numbered
%\bibliographystyle{model1-num-names}

%% Numbered without titles
%\bibliographystyle{model1a-num-names}

%% Harvard
%\bibliographystyle{model2-names.bst}\biboptions{authoryear}

%% Vancouver numbered
%\usepackage{numcompress}\bibliographystyle{model3-num-names}

%% Vancouver name/year
%\usepackage{numcompress}\bibliographystyle{model4-names}\biboptions{authoryear}

%% APA style
%\bibliographystyle{model5-names}\biboptions{authoryear}

%% AMA style
%\usepackage{numcompress}\bibliographystyle{model6-num-names}

%% `Elsevier LaTeX' style
%\bibliographystyle{elsarticle-num}
%%%%%%%%%%%%%%%%%%%%%%%

\title{A generalization of quantum pair state transfer}

\author{Sooyeong Kim\textsuperscript{1}, Hermie Monterde\textsuperscript{2}, Bahman Ahmadi\textsuperscript{3}, Ada Chan\textsuperscript{1}\\
	Stephen Kirkland,\textsuperscript{2} and  Sarah Plosker\textsuperscript{2,4}
}

\maketitle

\begin{abstract}
An \textit{$s$-pair state} in a graph is a quantum state of the form $\mathbf{e}_u+s\mathbf{e}_v$, where $u$ and $v$ are vertices in the graph and $s$ is a non-zero complex number. If $s=-1$ (resp., $s=1$), then such a state is called a \textit{pair state} (resp. \textit{plus state}). In this paper, we develop the theory of perfect $s$-pair state transfer in continuous quantum walks, where the Hamiltonian is taken to be the adjacency, Laplacian or signless Laplacian matrix of the graph. We characterize perfect $s$-pair state transfer in complete graphs, cycles and antipodal distance-regular graphs admitting vertex perfect state transfer. We construct infinite families of graphs with perfect $s$-pair state transfer using quotient graphs and graphs that admit fractional revival. We provide necessary and sufficient conditions such that perfect state transfer between vertices in the line graph relative to the adjacency matrix is equivalent to perfect state transfer between the plus states formed by corresponding edges in the graph relative to the signless Laplacian matrix. Finally, we characterize perfect state transfer between vertices in the line graphs of Cartesian products relative to the adjacency matrix.
\end{abstract}

	
\noindent \textbf{MSC2010 Classification:} {05C50; %Graphs and linear algebra (matrices, eigenvalues, etc.)
	81P45; %Quantum information, communication, networks
	%15A16; %Matrix exponential
	%05C22; %Signed and weighted graphs
	05C76; %Graph operations (line graphs, products, etc)
	15A18; %Eigenvalues, singular values, and eigenvectors
	81Q10 %Selfadjoint operator theory in quantum theory, including spectral analysis
 
\addtocounter{footnote}{1}
\footnotetext{Department of Mathematics and Statistics, York University, Toronto, ON, Canada M3J 1P3}
\addtocounter{footnote}{1}
\footnotetext{Department of Mathematics, University of Manitoba, Winnipeg, MB, Canada R3T 2N2}
\addtocounter{footnote}{1}
\footnotetext{Department of Mathematics, Shiraz University, Shiraz, Iran}
\addtocounter{footnote}{1}
\footnotetext{Department of Mathematics \& Computer Science, Brandon University, Brandon, MB, Canada R7A 6A9}

\medskip

\section{Introduction}
\label{Section:Introduction}

The use of a continuous-time quantum walk to transfer quantum states was proposed by Bose in 2003 \cite{Bose:Quantum}. Since then, continuous-time quantum walks have become invaluable tools in the theory of quantum computation and information. See \cite{Coutinho2021} for the background on continuous-time quantum walks.

Motivated by high probability quantum transmission, Christandl et. al\ introduced the concept of perfect state transfer in 2004 \cite{Christandl:PSTonHypercubes}. For two decades, the focus of most studies on perfect state transfer was between vertex states. However, a result due to Godsil implies that perfect state transfer between vertex states is rare \cite{Godsil2010}. This prompted Chen to extend the study of perfect state transfer to edge states \cite{chen2019edge}. Chen and Godsil subsequently expanded the preceding work to cover the so-called pair states and plus states \cite{Chen2020PairST}. In this paper, we study perfect state transfer between $s$-pair states, which is a natural generalization of both pair states and plus states.

Let $X$ be a connected graph on $n$ vertices, and $M$ be a Hermitian matrix associated with $X$.  
The {\sl continuous-time quantum walk on $X$ with Hamiltonian $M$} has transition matrix
\begin{equation*}
U_M(t) :=e^{-\ii t M}.
\end{equation*}
The transition matrix is unitary because $M$ is Hermitian.

Some results in this paper apply specifically to $M$ being
the adjacency matrix $A$ of $X$, the Laplacian matrix $L$ of $X$,
or the signless Laplacian matrix $Q$ of $X$. 
If we do not specify the Hamiltonian in a statement by using $U(t)$ to denote the
transition matrix, then one may assume it is any real symmetric matrix associated with $X$.
Unless explicitly stated, we assume $X$ is a simple connected undirected graph and all edges in $X$ are unweighted.

A {\sl (pure) state} is a $1$-dimensional subspace of $\C^{n}$.   We represent a state by a unit vector $\bu$ spanning the $1$-dimensional subspace.  
Note that $\eta\bu$ represents the same state as $\bu$, for any phase factor $\eta$.  The {\sl density matrix} of this state is $D_{\bu}:=\bu \bu^*$.  
We say the state $\bu$ is a {\sl real state} if its density matrix $D_\bu$ is real.

{\sl Perfect state transfer} occurs from the state $\bu$ to the state $\bmu$ at time $\tau$ if 
\begin{equation*}
U(\tau)\bu=\eta \bmu
\end{equation*} 
for some phase factor $\eta$, or equivalently,
\begin{equation*}
U(\tau)D_{\bu} U(-\tau) = D_{\bmu}.
\end{equation*}
A {\sl vertex state} is the characteristic vector $\be_a$ of some vertex $a$ in $X$.
We define an {\sl $s$-pair state} as a state in the form
\begin{equation*}
\frac{1}{\sqrt{1+\vert s\vert^2}}\left(\be_a+s \be_b\right)
\end{equation*}
for some non-zero complex number $s$ and distinct vertices $a$ and $b$. We say $\bu$ is a {\sl pair state} if $s=-1$, and it is a {\sl plus state} if $s=1$. For simpler exposition, we will drop the normalization factor $\frac{1}{\sqrt{1+s^2}}$.

%An $s$-pair state, simply put, is a pair of entangled qubits in a quantum spin network, where $s$ represents the strength of the entanglement between the two qubits. In particular, if $\vert s\vert$ is close to 1, then an $s$-pair state exhibits a higher degree of entanglement. Thus, we may think of pair states (resp., plus states) as $s$-pair states that are maximally entangled.

When studying a quantum spin network, one considers an arbitrary graph where each vertex represents a spin, and the weight of the edge between vertices represents the coupling strength of interaction between the two spins in the quantum system. Strictly speaking, the vector $\be_a$ represents excitation of spin $a$, and so the state of the system can be represented as $\ket{1}_a\ket{0}_b\ket{0}_c\cdots\in \mathbb C^{2^n}$. Similarly,  $\be_b$ corresponds to $\ket{0}_a\ket{1}_b\ket{0}_c\cdots$, where we have, without loss of generality, labelled our vertices so that the vertices $a$ and $b$ appear first. Since the excitation only occurs on vertices $a$ and $b$, we can effectively ignore the rest of the system (mathematically, we can trace out all other qubits) and focus solely on the state of two qubits: $\ket{1}_a\ket{0}_b+s\ket{0}_a\ket{1}_b$, where $s\in \mathbb C$ and we have again dropped the normalization factor for simplicity. In this way, an $s$-pair state $\be_a+s\be_b$ represents a pair of entangled qubits, forming a state in the 1-excitation subspace $\mathbb C^n$ of the full $2^n$-dimensional system of $n$ spins. This state is always entangled for any $s\neq 0$.

Furthermore, although $s$ is not in general a Schmidt coefficient (since we are allowing for $s\in \mathbb C$), we can still view $s$ as a measure of the degree of entanglement of our state $\ket{1}_a\ket{0}_b+s\ket{0}_a\ket{1}_b$, in so far as the  entropy of entanglement is given by
\begin{eqnarray*}
	S(\rho_a)&=&S\left(\tr_b (\ket{1}_a\ket{0}_b+s\ket{0}_a\ket{1}_b)(\bra{1}_a\bra{0}_b+\bar{s}\bra{0}_a\bra{1}_b)\right)\\
	&=&S(\ketbra{1_a}{1_a}+s\ketbra{0_a}{1_a}+\bar{s}\ketbra{1_a}{0_a}+\lvert s\rvert^2\ketbra{0_a}{0_a}),    
\end{eqnarray*}
where $\tr_b$ is the partial trace over the $b$ subsystem. The closer $\vert s\vert$ is to zero, the closer the density matrix $\rho=(\ket{1}_a\ket{0}_b+s\ket{0}_a\ket{1}_b)(\bra{1}_a\bra{0}_b+\bar{s}\bra{0}_a\bra{1}_b)$ is to $\ketbra{1_a}{1_a}$, a pure state (having entropy of zero), and therefore the closer the state $\ket{1}_a\ket{0}_b+s\ket{0}_a\ket{1}_b$ is to a separable state (namely $\ket{1}_a\ket{0}_b$). On the other hand, the closer $\vert s\vert$ is to 1, the closer $S(\rho_a)$ is to its maximum value, and the closer $\ket{1}_a\ket{0}_b$ is to a maximally entangled state. Thus, we may think of pair states (resp., plus states) as $s$-pair states that are maximally entangled.

We say that perfect state transfer between vertices $a$ and $b$ at time $\tau$ if
\begin{equation*}
U(\tau) \be_a = \eta \be_b,
\end{equation*}
for some phase factor $\eta$.   If the Hamiltonian is real and symmetric then $U(\tau)$ is symmetric and
$U(\tau) \be_b = \eta \be_a$, which gives
\begin{equation*}
U(\tau) (\be_a+s\be_b) = \eta (\be_b+s \be_a),\quad \text{for $s\in \C$.}
\end{equation*}
Hence, we can view perfect state transfer between vertices as a special case of perfect transfer of $s$-pair states.

Given two states 
$\bu=\be_a+r \be_b$ and $\bmu=\be_\alpha +s \be_\beta$, 
we say {\sl perfect $s$-pair state transfer} occurs from $\bu$ to $\bmu$  if 
$U(\tau)\bu=\eta \bmu$, 
for some time $\tau$. As a preliminary investigation of $s$-pair states, we opt to study perfect $s$-pair state transfer 
where $r=s$ is a non-zero real number. 
\ignore{
	satisfying equations of the form
	\begin{equation*}
	U(\tau) \left(\frac{1}{\sqrt{1+s^2}}\left(\be_a + s \be_b\right)\right) = \eta  \left(\frac{1}{\sqrt{1+s^2}}\left(\be_\alpha + s \be_\beta\right)\right),
	\end{equation*}
	and restrict $s$ to be non-zero real number.}
When $s=\pm 1$, we have perfect pair state transfer and perfect plus state transfer as introduced in \cite{Chen2020PairST}. Since perfect state transfer between $s$-pair states represents accurate transmission of a pair of entangled qubits to another pair of entangled qubits in a quantum spin network, it follows that perfect $s$-pair state transfer allows for the transfer and generation of entanglements, a property considered desirable in quantum information theory \cite{Chan2019}. 
This paper is organized as follows. In Section~\ref{Section:P2ST}, we provide necessary conditions for perfect $s$-pair state transfer and supply examples of graphs that admit perfect $s$-pair state transfer. Section~\ref{Section:Periodic} deals with real periodic $s$-pair states. In particular, we prove that for every positive integer $k$ and positive rational $s$, there are only finitely many connected graphs with maximum valency $k$ such that $\be_a+s\be_b$ is periodic in $X$ relative to the adjacency, Laplacian and signless Laplacian matrix. Thus, similar to the vertex case \cite{Godsil2010}, perfect $s$-pair state transfer is rare when $s$ is a positive rational number. In Section~\ref{Section:SC}, we establish combinatorial and algebraic properties of graphs with strongly cospectral $s$-pair states.
Section~\ref{Section:Vto2} is dedicated to constructions of graphs with perfect $s$-pair state transfer using quotient graphs and graphs that admit fractional revival. We also extend a transitivity property of perfect pair state transfer in \cite{Chen2020PairST} to perfect $s$-pair state transfer. In Section~\ref{Section:Special}, we characterize perfect $s$-pair state transfer in complete graphs and cycles. It turns out that complete graphs do not admit perfect $s$-pair state transfer, while $C_4$, $C_6$ and $C_8$ are the only cycles that admit perfect $s$-pair state transfer. For distance-regular graphs admitting perfect state transfer between vertices, we provide necessary and sufficient conditions such that these graphs also admit perfect $s$-pair state transfer. Section~\ref{Section:Line} is devoted to exploring the relationship between the existence of perfect state transfer between plus states formed by edges in a graph relative to the signless Laplacian matrix, and the existence of perfect state transfer between the corresponding vertices in the line graph relative to the adjacency matrix. Then we utilize the singular values and singular vectors of the incidence matrix of a graph to characterize strong cospectrality and perfect vertex state transfer in the line graph. Finally, in Section \ref{Section:Cart}, we characterize adjacency perfect state transfer between vertices in the line graphs of Cartesian products. Taken together, our results broaden the literature on pair and plus states, establishing new instances of perfect state transfer between $s$-pair states, while developing techniques that will facilitate future research on this topic.

%%%%%%%%%%%%%%%%%%%%%%%%%%%%%%%%%%%%%%%%%%%%%%%%%%%%%%%%%%%%%%%%%%%%%%%%%%%%%%
%%%%%%%%%%%%%%%%%%%%%%%%%%%%%%%%%%%%%%%%%%%%%%%%%%%%%%%%%%%%%%%%%%%%%%%%%%%%%%
%%%%%%%%%%%%%%%%%%%%%%%%%%%%%%%%%%%%%%%%%%%%%%%%%%%%%%%%%%%%%%%%%%%%%%%%%%%%%%
%%%%%%%%%%%%%%%%%%%%%%%%%%%%%%%%%%%%%%%%%%%%%%%%%%%%%%%%%%%%%%%%%%%%%%%%%%%%%%
%%%%%%%%%%%%%%%%%%%%%%%%%%%%%%%%%%%%%%%%%%%%%%%%%%%%%%%%%%%%%%%%%%%%%%%%%%%%%%

\section{Perfect $s$-pair state transfer}
\label{Section:P2ST}

In a graph $X$, {\sl perfect $s$-pair state transfer} occurs from $\bu = \be_a + s \be_b$ to $\bmu=\be_{\alpha} + s \be_{\beta}$ at time $\tau$ if
there exists a unit complex number $\eta$, called a {\sl phase factor}, such that
\begin{equation}
\label{Eqn:P2T}
U(\tau) \bu = \eta \bmu,
\end{equation}
equivalently,
\begin{equation}
\label{Eqn:P2TDensity}
U(\tau) D_\bu U(-\tau) = D_\bmu.
\end{equation}
If $\bu=\bmu$ then we say the state $\bu$ is {\sl periodic} at time $\tau$.
Different from vertex states, it is possible for an $s$-pair state $\bnu$ to be an eigenvector of $M$ corresponding to some eigenvalue $\lambda$.  In this case, 
\begin{equation*}
U(t) \bnu = e^{-\ii t \lambda} \bnu,
\end{equation*}
for any time $t$, and we call $\bnu$ a {\sl fixed state}.

\begin{example}
	\label{Eg:P2ST}
	We give  a weighted graph $P_5(w)$ with perfect $s$-pair state transfer, and an infinite family of trees admitting pair state transfer.
	\begin{enumerate}[(a)]
		\item
		\label{Eg:P5}
		For a positive real number $w$, the weighted path $P_5(w)$ has perfect $s$-pair state transfer from $\left(\be_3-\frac{2}{\sqrt{w}}\be_1\right)$ to $\left(\be_3-\frac{2}{\sqrt{w}}\be_5\right)$ at time $\frac{\pi}{\sqrt{w}}$.  Note that $s\neq \pm 1$ for $w\neq 4$.
		
		\begin{figure}[h!]
			\begin{center}
				\includegraphics[scale=0.8]{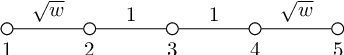}
				\caption{The weighted path $P_5(w)$} \label{Fig:P5}
			\end{center}
		\end{figure}
		\item
		\label{Eg:Pal}
		As a special case of the  construction in \cite{Pal2024},
		the infinite family of trees $T_n$ shown in Figure~\ref{Fig:Pal}, with $n\geq 0$, has adjacency perfect $s$-pair state transfer between two states $\be_a-\be_b$ and $\be_\alpha-\be_\beta$ at time $\frac{\pi}{\sqrt{2}}$.
		\begin{figure}[h!]
			\begin{center}
				\includegraphics[scale=0.8]{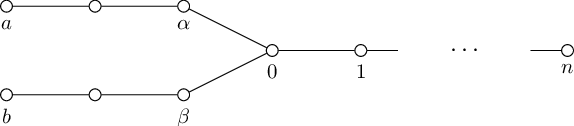}
				\caption{$T_n$} \label{Fig:Pal}
			\end{center}
		\end{figure}
	\end{enumerate}
\end{example}
\vspace{-0.2in}

For a Hamiltonian $M$, we use $\spec(M)$ to denote its spectrum.
Given the spectral decomposition of the Hamiltonian
\begin{equation}
\label{Eqn:SpecDecomp}
M = \sum_{\lambda \in \spec(M)} \lambda E_\lambda,
\end{equation}
we have 
\begin{equation*}
U_M(t) =  \sum_{\lambda \in \spec(M)} e^{-\ii t \lambda } E_\lambda.
\end{equation*}  
Multiplying $E_\lambda$ to both sides of Equation~(\ref{Eqn:P2T}) yields
\begin{equation*}
E_\lambda \bu = \left( e^{\ii \tau \lambda}\eta\right) E_\lambda \bmu,
\quad \text{for $\lambda \in \spec(M)$.}
\end{equation*}
Since $E_\lambda$, $\bu$ and $\bmu$ are real,  $e^{\ii \tau \lambda}\eta$ is a real phase factor.
Hence we have 
\begin{equation}
\label{Eqn:SC}
E_\lambda \bu =\pm E_\lambda \bmu, \quad \text{for $\lambda \in \spec(M)$.}
\end{equation}
We say the $s$-pair states are {\sl strongly cospectral} if they satisfy the above condition.  
As in the case of vertex state transfer, strong cospectrality is a necessary condition for perfect $s$-pair state transfer.

For an arbitrary state $\bnu \in \C^n$,  the {\sl eigenvalue support of $\bnu$} relative to $M$ is the set
\begin{equation*}
\supp_{\bnu} := \left\{ \lambda : E_\lambda\bnu \neq \zero\right\}.
\end{equation*}
%If $\bnu=\be_a$ for some vertex $a$, then we abbreviate $\supp_{\be_a}$ by $\supp_a$.
It is obvious that if $\bu=\be_a+s\be_b$ then $\supp_\bu \subseteq \supp_{\be_a} \cup \supp_{\be_b}$.

Suppose $\bu$ and $\bmu$ are strongly cospectral states.  Equation~(\ref{Eqn:SC}) implies $\supp_\bu=\supp_\bmu$, and gives the  natural partition
of $\supp_\bu = \supp_{\bu, \bmu}^+\ \dot\cup\ \supp_{\bu, \bmu}^-$, where
\begin{equation*}
\supp_{\bu, \bmu}^+ = \left\{\lambda: E_\lambda \bu = E_\lambda \bmu \neq \zero \right\}
\quad \text{and} \quad
\supp_{\bu, \bmu}^- = \left\{\theta : E_\theta \bu = -E_\theta \bmu \neq \zero\right\}.
\end{equation*}

We now give a lower bound on the size of the support of an $s$-pair state $\be_a+s\be_b$ in terms of the distance of $a$ and $b$, denoted by $\dist(a,b)$.

\begin{proposition}
	\label{Prop:Support}
	Suppose the Hamiltonian for a graph $X$ is either $A$, $L$ or $Q$.
	If $\bu=\be_a + s \be_b$ is a fixed state then $\vert \supp_\bu\vert = 1$.  Otherwise, 
	\begin{equation*}
	\vert \supp_\bu \vert \geq \left\lceil \frac{\dist(a,b)}{2} \right\rceil.
	\end{equation*}
	\begin{proof}
		Let $k$ be the maximum degree of $X$ and let $M$ be $I+A$, $(k+1)I-L$ or $Q$ if the Hamiltonian  is $A$, $L$ or $Q$, respectively.
		For each case, the support of $\bu$ with respect to $M$ is the same as that with respect to the original Hamiltonian.
		Note that $M$ is a non-negative matrix with positive diagonal entries.
		
		If $M\bu = \lambda \bu$ then $\supp_\bu=\{\lambda\}$. On the other hand, if $\bu$ is not an eigenvector of $M$, then the  number of non-zero entries in the vectors $M^\ell \bu$ is strictly increasing as $\ell$ increases from $0$ to $\left(\left\lceil \frac{\dist(a,b)}{2} \right\rceil-1\right)$.
		Hence the set
		\begin{equation*}
		\left\{M^\ell \bu : 0\leq \ell \leq \left\lceil \frac{\dist(a,b)}{2} \right\rceil-1\right\}
		\end{equation*}
		is linearly independent in $\spn\{E_\lambda\bu: j=1,\ldots, d\}$ which has dimension $\vert \supp_\bu\vert$.
	\end{proof}
\end{proposition}

Adapting the proof of Theorems~2.4.2 to 2.4.4 in \cite{Coutinho2014} yields the following characterization of perfect state transfer between real states.
\begin{theorem}
	\label{Thm:RealPST}
	%Let $M$ be a real and symmetric Hamiltonian associated with $X$.  
	Let $\bu$ and $\bmu$ be real states, and $\supp_\bu$ be closed under taking algebraic conjugates.
	Perfect state transfer occurs from $\bu$ to $\bmu$ if and only if the following conditions hold.
	\begin{enumerate}[(i)]
		\item
		\label{Thm:RealPST1}
		The $s$-pair states $\bu$ and $\bmu$ are strongly cospectral.
		\item
		\label{Thm:RealPST2}
		The elements in $\supp_{\bu}=\supp_{\bmu}$ are either
		\begin{enumerate}[(a)]
			\item
			all integers, or
			\item
			there exists a square-free integer $\Delta > 1$ and an integer $c$ such that each element in $\supp_\bu$ is in the form
			$\frac{c+d\sqrt{\Delta}}{2}$, for some integer $d$.
		\end{enumerate}
		\item
		\label{Thm:RealPST3}
		Let $\lambda \in  \supp_{\bu, \bmu}^+$, and
		\begin{equation*}
		g=\gcd \left\{\frac{\lambda - \theta}{\sqrt{\Delta}}\right\}_{\theta \in \supp_\bu}
		\end{equation*}
		(with $\Delta=1$ for Case~(a) above).
		Then $\theta \in \supp_{\bu, \bmu}^+$ if and only if $\frac{\lambda - \theta}{g\sqrt{\Delta}}$ is even.
		
	\end{enumerate}
	The minimum perfect $s$-pair state transfer time is $\frac{\pi}{g\sqrt{\Delta}}$, and $U\left(\frac{\pi}{g\sqrt{\Delta}}\right)\bu = e^{-\ii \tau \lambda}\bmu$.
\end{theorem}

If $\bu$ is a fixed state relative to $M$, then $\vert\supp_\bu\vert=1$ by Proposition \ref{Prop:Support}, and so $\bu$ is not involved in strong cospectrality. In this case, $\bu$ cannot exhibit perfect state transfer relative to $M$ by Theorem \ref{Thm:RealPST} (\ref{Thm:RealPST1}). By way of example, if $a$ and $b$ are twin vertices in $X$ ( i.e., $a$ and $b$ have the same neighbours), then $\be_a-\be_b$ is an eigenvector for $M\in\{A,L,Q\}$ \cite{Monterde2022}. Thus, $\be_a-\be_b$ is a pair state in $X$ that is not involved in perfect state transfer relative to $M\in\{A,L,Q\}$ .

Note that the following proposition applies to $M\in \{A, L, Q\}$
\begin{proposition}
	\label{Prop:Rat}
	Let $\bu=\be_a+s\be_b$, for some $s\in \Q\backslash\{0\}$.
	If the entries of $M$ are algebraic integers, then $\supp_\bu$ is closed under taking algebraic conjugates.
	\begin{proof}
		If $\lambda$ and $\theta$ are algebraic conjugates, then so are $E_{\lambda}$ and $E_{\theta}$.
		Therefore,   $E_{\lambda} \bu = \zero$ if and only if $E_{\theta} \bu = \zero$.
	\end{proof}
\end{proposition}

When $s\in \R\backslash\{0\}$, $D_\bu$ and $D_\bmu$ are real density matrices and the following theorem follows from
Lemmas~2.3, 5.2 and Corollary 5.3 of \cite{Godsil2017RST}.
\begin{theorem}
	\label{Thm:RST}
	Let $\bu = \be_a + s \be_b$ and $\bmu=\be_{\alpha} + s \be_{\beta}$ be distinct $s$-pair states, for some $s\in \R\backslash\{0\}$.
	If perfect $s$-pair state transfer occurs from $\bu$ to $\bmu$ at time $\tau$, then
	\begin{enumerate}[(i)]
		\item
		\label{Thm:RST1}
		Perfect $s$-pair state transfer occurs from $\bmu$ to $\bu$ at time $\tau$.
		\item
		\label{Thm:RST2}
		Both $\bu$ and $\bmu$ are periodic at time $2\tau$.
		\item
		\label{Thm:RST3}
		If the minimum period of $\bu$ is $\tau$ then perfect $s$-pair state transfer between $\bu$ and $\bmu$ occurs at time $\frac{\tau}{2}$.
		\item
		\label{Thm:RST4}
		There is no perfect state transfer from $\bu$ to a state with real density matrix other than $D_\bu$ and $D_\bmu$.
	\end{enumerate}
\end{theorem}

%%%%%%%%%%%%%%%%%%%%%%%%%%%%%%%%%%%%%%%%%%%%%%%%%%%%%%%%%%%%%%%%%%%%%%%%%%%%%%
%%%%%%%%%%%%%%%%%%%%%%%%%%%%%%%%%%%%%%%%%%%%%%%%%%%%%%%%%%%%%%%%%%%%%%%%%%%%%%
%%%%%%%%%%%%%%%%%%%%%%%%%%%%%%%%%%%%%%%%%%%%%%%%%%%%%%%%%%%%%%%%%%%%%%%%%%%%%%
%%%%%%%%%%%%%%%%%%%%%%%%%%%%%%%%%%%%%%%%%%%%%%%%%%%%%%%%%%%%%%%%%%%%%%%%%%%%%%
%%%%%%%%%%%%%%%%%%%%%%%%%%%%%%%%%%%%%%%%%%%%%%%%%%%%%%%%%%%%%%%%%%%%%%%%%%%%%%

\section{Periodic $s$-pair states}
\label{Section:Periodic}

From Thereom~\ref{Thm:RST}, a  real $s$-pair state is necessarily periodic if it is involved in perfect $s$-pair state transfer in $X$. 
We focus here on the real periodic $s$-pair states.

A set $S \subset \R$ with at least two elements satisfies the {\sl ratio condition} if
for any $\lambda_h, \lambda_j, \lambda_k, \lambda_l \in S$ with $\lambda_k\neq \lambda_l$,
\begin{equation*}
\frac{\lambda_h-\lambda_j}{\lambda_k-\lambda_l} \in \Q.
\end{equation*}
If $\vert S\vert=2$, then $S$ automatically satisfies the ratio condition. 
The following theorem follows  directly from Theorem 9.1.1 in \cite{Coutinho2021}.

%Corollary~7.6.2, and Theorems 9.1.1 and 9.5.1 in \cite{Coutinho2021}.

\begin{theorem}
	\label{Thm:RatioCond}
	Let $s\in \R\backslash\{0\}$.  The $s$-pair state $\bu = \be_a + s \be_b$ is periodic in $X$ if and only if $\supp_\bu$ satisfies the ratio condition.
\end{theorem}

By applying Theorem~7.6.1 of \cite{Coutinho2021}, we obtain further restrictions on the eigenvalues if $\supp_\bu$ is closed under taking algebraic
conjugates.
\begin{theorem}
	\label{Thm:Quadratic}
	Let $\bu = \be_a + s \be_b$, for some $s\in \R\backslash\{0\}$.  
	If $\vert \supp_\bu \vert > 2$ and $\supp_\bu$ is closed under taking algebraic conjugates, then $\supp_\bu$ satisfies the ratio
	condition  if and only if one of the following holds.
	\begin{enumerate}[(i)]
		\item
		The elements in $\supp_\bu$ are integers.
		\item
		There is a square-free integer $\Delta > 1$ and an integer $c$ such that each element in $\supp_\bu$ is in the form
		\begin{equation*}
		\frac{c+d\sqrt{\Delta}}{2},\quad  \text{for some integer $d$.}
		\end{equation*}
	\end{enumerate}
	If either condition holds (with $\Delta=1$ in the first case), let
	\begin{equation*}
	g=\gcd \left\{\frac{\lambda - \theta}{\sqrt{\Delta}}\right\}_{\lambda, \theta \in \supp_\bu}.
	\end{equation*}
	The minimum period of $\bu$ is $\frac{2\pi}{g\sqrt{\Delta}}$,  and 
	\begin{equation*}
	U\left(\frac{2\pi}{g\sqrt{\Delta}}\right)\bu = e^{-\ii \tau \lambda}\bu, \quad \text{for $\lambda\in\supp_\bu$.}
	\end{equation*}
\end{theorem}

\begin{corollary}
	If $\supp_{\be_a} \cup \supp_{\be_b}$ satisfies the ratio condition then the $s$-pair state $\bu = \be_a + s \be_b$ is periodic for any $s\in \R\backslash\{0\}$. 
	
	\begin{proof}
		This result follows immediately from the fact that $\supp_\bu \subseteq \supp_{\be_a} \cup \supp_{\be_b}$.
	\end{proof}
\end{corollary}

In particular, if $U(\tau)\be_a = \eta \be_a$ and $U(\tau)\be_b = \eta \be_b$,
then  $\be_a + s \be_b$ is periodic, for any $s\in \R\backslash\{0\}$. 
We now give a family of graphs that have periodic pair states but no periodic vertices.
Let $X$ be a conference graph on $n$ vertices, where $\sqrt{n} \not\in \Z$.  (See Section~1.3 of \cite{BCN1989} for conference graphs.)
Note that $X$ is regular with valency $k=(n-1)/2$.  The eigenvalue support of each vertex with respect to the adjacency matrix consists of
\begin{equation*}
\lambda_1=k, \quad  \lambda_2= \frac{-1+\sqrt{n}}{2},\quad \text{and}\quad \lambda_3= \frac{-1-\sqrt{n}}{2},
\end{equation*}
which does not satisfy the ratio condition.
Hence $X$ has no periodic vertices.   Let $\bu=\be_a-\be_b$.  Since $E_{\lambda_1}=\frac{1}{n}J$, we have $E_{\lambda_1} \bu=\zero$ and $\supp_\bu = \{\lambda_2, \lambda_3\}$ satisfies the ratio
condition.  By Theorem~\ref{Thm:RatioCond}, the pair state $\bu$ is periodic.  Since $X$ is regular, $\bu$ is also periodic relative to
the Laplacian or the signless Laplacian matrix of $X$.  

The next corollary follows immediately from Theorem~\ref{Thm:Quadratic}.
\begin{corollary}
	\label{Cor:Gap}
	Let $\bu=\be_a + s \be_b$ be a real periodic $s$-pair state in $X$.   If $\supp_\bu$ is closed under taking algebraic conjugates, then $\vert \lambda - \theta \vert \geq 1$,
	for all $\lambda, \theta \in \supp_\bu$ with $\lambda\neq \theta$.
\end{corollary}

%Corollary~7.7.1 of \cite{Coutinho2021}

In \cite{Godsil2010}, Godsil showed that graphs with periodic vertices relative to the adjacency matrix are rare. We show a similar statement about real periodic $s$-pair states with non-negative entries. We denote the spectral radius of a square matrix $N$ by $\rho(N)$.

\begin{corollary}
	\label{Cor:MaxDeg}
	For each positive integer $k$ and positive rational number $s$, there are finitely many connected graphs $X$ with maximum valency at most $k$ such that $\be_a+s\be_b$ is periodic in $X$ with Hamiltonian $A$, $L$, or $Q$.
	\begin{proof}
		We let $M$ be $I+A$, $(k+1)I-L$ and $Q$ if the Hamiltonian  is $A$, $L$ or $Q$, respectively.
		Let the spectral decomposition of $M$ be
		\begin{equation*}
		M=\sum_{\lambda \in \spec(M)} \lambda E_\lambda.
		\end{equation*}
		The $E_\lambda$'s are also the orthogonal projection matrices onto the eigenspaces of the original Hamiltonian.
		For each case, an $s$-pair state $\bu= \be_a + s \be_b$ is periodic with $M$ being the Hamiltonian if and only if $\bu$ is periodic with the original Hamiltonian.
		
		Let $r$ be the covering radius of the set $\{a, b\}$ in $X$,  which is the smallest integer $r$ such that every vertex in $X$  is at distance at most $r$ from some vertex in $\{a,b\}$.  Since the entries in $M$ and $\bu$ are non-negative, 
		\begin{equation*}
		\{M^h\bu : h=0,\ldots, r\}
		\end{equation*}
		is a linearly independent set in  $\spn \{E_\lambda \bu : \lambda \in \supp_\bu\}$.  Hence $r < \vert \supp_\bu \vert$.
		
		Now, observe that
		\begin{equation*}
		\rho(I+A) = \rho(A)+1 \leq k+1.
		\end{equation*}
		For each vertex $v$, let $d(v)$ be the degree of $v$ in $X$.   The proof of Theorem~3.1 in \cite{shi2007Lbound} states the inequality
		\begin{equation*}
		\rho(L) \leq \rho(Q) \leq \sqrt{2} \max_{v} \sqrt{d(v)^2+\sum_{v' \sim v} d(v')} \leq 2k.
		\end{equation*}
		We conclude that $\rho\left((k+1)I-L\right) \leq k+1$ and $\rho(Q) \leq 2k$.
		
		For all three choices of the Hamiltonian, Corollary~\ref{Cor:Gap} implies 
		\begin{equation*}
		r<\vert \supp_\bu \vert \leq 2\rho(M)+1 \leq 4k+1.
		\end{equation*}
		For each positive integer $k$, there are finitely many connected graphs with maximum degree at most $k$
		and covering radius of a pair of vertices bounded above by $4k$.
	\end{proof}
\end{corollary}

\begin{remark}
	\quad 
	\begin{enumerate}[(a)]
		\item
		The proof of the Corollary~\ref{Cor:MaxDeg} is adaptable to any  state $\bu$ with non-negative rational entries.  In particular, it applies to periodic vertex states.
		\item
		Corollary~\ref{Cor:MaxDeg} does not hold when $s<0$.   
		Figure~\ref{Fig:Pal} gives an infinite family of trees $\left\{T_n: n \geq 0\right\}$ with maximum degree three containing periodic $s$-pair state $\be_a-\be_b$.
	\end{enumerate}
\end{remark}

We close this section with a result that will prove useful in Section~\ref{Subsection:Cn}.
\begin{proposition}
	\label{Prop:CospectralPeriodic}
	Let $a$ and $b$ be cospectral vertices in $X$.   If $(\be_a + s \be_b)$ is periodic at time $\tau$, then either $s=\pm 1$ or both $a$ and $b$ are periodic at time $\tau$.
	\begin{proof}
		Suppose $U(\tau) (\be_a + s \be_b) = \eta(\be_a + s \be_b)$ for some phase factor $\eta$, which gives
		\begin{equation*}
		\begin{cases}
		\be_a^TU(\tau)\be_a + s \be_a^T U(\tau) \be_b = \eta\\
		\be_b^TU(\tau)\be_a + s \be_b^T U(\tau) \be_b =s \eta.
		\end{cases}
		\end{equation*}
		We have $\be_a^T U(\tau)\be_a = \be_b^TU(\tau)\be_b$ because $a$ and $b$ are cospectral,  and $\be_b^T U(\tau)\be_a = \be_a^TU(\tau)\be_b$ because $U(\tau)$ is symmetric.
		The above system of equations give
		\begin{equation*}
		(s^2-1) \be_a^TU(\tau)\be_b =0.
		\end{equation*}
		Hence either $s=\pm 1$, or $\be_a^TU(\tau)\be_a = \be_b^T U(\tau) \be_b = \eta$.
	\end{proof}
\end{proposition}
%%%%%%%%%%%%%%%%%%%%%%%%%%%%%%%%%%%%%%%%%%%%%%%%%%%%%%%%%%%%%%%%%%%%%%%%%%%%%%
%%%%%%%%%%%%%%%%%%%%%%%%%%%%%%%%%%%%%%%%%%%%%%%%%%%%%%%%%%%%%%%%%%%%%%%%%%%%%%
%%%%%%%%%%%%%%%%%%%%%%%%%%%%%%%%%%%%%%%%%%%%%%%%%%%%%%%%%%%%%%%%%%%%%%%%%%%%%%
%%%%%%%%%%%%%%%%%%%%%%%%%%%%%%%%%%%%%%%%%%%%%%%%%%%%%%%%%%%%%%%%%%%%%%%%%%%%%%
%%%%%%%%%%%%%%%%%%%%%%%%%%%%%%%%%%%%%%%%%%%%%%%%%%%%%%%%%%%%%%%%%%%%%%%%%%%%%%

\section{Strongly cospectral $s$-pair states}
\label{Section:SC}

In this section, we make some observations on strongly copsectral $s$-pair states, a necessary condition for perfect $s$-pair state transfer as stated in Theorem~\ref{Thm:RealPST}.

Given the spectral decomposition $\displaystyle M=\sum_{\lambda \in \spec(M)} \lambda E_\lambda$, the real states $\bu$ and $\bmu$ are strongly cospectral with respect to $M$
implies 
\begin{equation*}
\bu^T E_\lambda \bu = \bmu^T E_\lambda \bmu
\end{equation*}
for all $j$, which is equivalent to
\begin{equation}
\label{Eqn:SCWalkM}
\bu^T M^h \bu = \bmu^T M^h \bmu,
\end{equation}
for all non-negative integer $h$.

\begin{proposition}
	\label{Prop:ASC}
	Let $\bu = \be_a + r \be_b$ and $\bmu = \be_\alpha+s \be_\beta$ be strongly cospectral $s$-pair states with respect to the adjacency matrix of $X$,
	for some real numbers $r$ and $s$.
	If $r\neq s$ then $a$ is not adjacent to $b$ and $\alpha$ is not adjacent to $\beta$ in $X$.
	\begin{proof}
		Applying (\ref{Eqn:SCWalkM}) with $h=1$ gives
		$r A_{a,b} = s A_{\alpha,\beta}$.
		If $r\neq s$ then 
		$A_{a,b} = A_{\alpha,\beta}=0$.
	\end{proof}
\end{proposition}

\begin{remark}
	Theorem~8.1 of \cite{Chen2020PairST} rules out Laplacian perfect $s$-pair state transfer from $\be_a+\be_b$ to $\be_\alpha-\be_\beta$.  Proposition~\ref{Prop:ASC} rules out adjacency  perfect $s$-pair state transfer from 
	$\be_a+\be_b$ to $\be_\alpha-\be_\beta$  if either $a$ is adjacent to $b$ or $\alpha$ is adjacent to $\beta$.
\end{remark}

\begin{proposition}
	\label{Prop:Walks}
	Suppose $\bu=\be_a+s \be_b$ is strongly cospectral with $\bmu = \be_{\alpha} + s \be_{\beta}$ with respect to $M$,
	for some $s\in \R\backslash\{0\}$.
	Then one of the following holds.
	\begin{enumerate}[(i)]
		\item
		\label{Prop:Walks1}
		$(M^h)_{a,a}=(M^h)_{\alpha, \alpha}$, $(M^h)_{b,b}=(M^h)_{\beta,\beta}$, and $(M^h)_{a,b}=(M^h)_{\alpha,\beta}$, for all $h\geq 0$.
		\item
		\label{Prop:Walks2}
		$s$ is the root of a polynomial of degree at most two.
	\end{enumerate}
	\begin{proof}
		Equation~(\ref{Eqn:SCWalkM}) gives 
		\begin{equation*}
		s^2 \left((M^h)_{b,b}-(M^h)_{\beta,\beta}\right) + 2s\left((M^h)_{a,b}-(M^h)_{\alpha,\beta}\right) + (M^h)_{a,a}-(M^h)_{\alpha,\alpha}=0,
		\end{equation*}
		for all $h\geq 0$.
		Either Condition~(\ref{Prop:Walks1}) holds, or $s$ is a root of some polynomial of degree at most two.
	\end{proof}
\end{proposition}

Theorem~8.3 of \cite{Chen2020PairST} relates Laplacian perfect pair state transfer to perfect plus state transfer with respect to $Q$ for bipartite graphs.
In the following proposition, we give a similar result for strong cospectrality of $s$-pair states in bipartite graphs.
\begin{proposition}
	\label{Prop:Bip}
	Let $X$ be a bipartite graph with bipartition $B_1$ and $B_2$.  Let $a, \alpha \in B_1$ and $b, \beta \in B_2$.
	Then the two states $(\be_a+s \be_b)$ and $(\be_\alpha+s\be_\beta)$ are strongly cospectral with respect to $Q$ if and only if
	$(\be_a-s \be_b)$ and $(\be_\alpha-s\be_\beta)$ are strongly cospectral with respect to $L$.
	\begin{proof}
		Let $P$ be the diagonal matrix where
		\begin{equation*}
		P_{v,v} =
		\begin{cases} 1 & \text{if $v\in B_1$,}\\
		-1& \text{if $v\in B_2$.}
		\end{cases}
		\end{equation*}   
		The proposition follows from
		\begin{equation*}
		U_L(t) P = PU_Q(t),
		\end{equation*}
		and $P(\be_{v_1}+s\be_{v_2})=\be_{v_1}-s\be_{v_2}$, for $v_1\in B_1$ and $v_2 \in B_2$.
	\end{proof}
\end{proposition}
Proposition~\ref{Prop:Bip} does not hold when $a, b \in B_1$.    In Figure~\ref{Fig:Bipartite}, the states 
$(\be_a+ \be_b)$ and $(\be_\alpha+ \be_\beta)$ are strongly cospectral with respect to $Q$ but
$(\be_a- \be_b)$ and $(\be_\alpha-\be_\beta)$ are not strongly cospectral with respect to $L$.

\begin{figure}[h!]
	\begin{center}
		\includegraphics[scale=0.8]{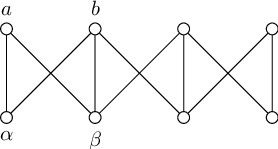}
		\caption{A counter-example of Proposition~\ref{Prop:Bip} if $a,b\in B_1$} \label{Fig:Bipartite}
	\end{center}
\end{figure}

We generalize Corollary~6.4 of \cite{GodsilSmith2024} to $s$-pair states.
\begin{proposition}
	\label{Prop:Aut}
	Let $\bu=\be_a+s \be_b$ and $\bmu=\be_\alpha+s \be_\beta$ be strongly cospectral $s$-pair states in $X$.
	If $s\neq 1$ then any automorphism of $X$ that fixes $(a,b)$ also fixes $(\alpha, \beta)$.
	If $s=1$ then any automorphism of $X$ that fixes $\{a,b\}$ also fixes $\{\alpha, \beta\}$.
	\begin{proof}
		Let $P$ be the permutation matrix of  an automorphism that fixes the state $\bu$. 
		Then $P$ commutes with each $E_j$.
		It follows from Equation~(\ref{Eqn:SC}) that 
		\begin{equation*}
		\bmu = \sum_{\lambda \in \supp_\bmu} E_\lambda \bmu = \left( \sum_{\lambda \in \supp_{\bu,\bmu}^+} E_\lambda - \sum_{\theta \in \supp_{\bu,\bmu}^-} E_\theta \right) \bu.
		\end{equation*}
		Hence
		\begin{equation*}
		P\bmu =  \left( \sum_{\lambda \in \supp_{\bu,\bmu}^+} E_\lambda - \sum_{\theta \in \supp_{\bu,\bmu}^-} E_\theta\right)P \bu,
		\end{equation*}
		and $P\bu = \bu$ if and only if $P\bmu= \bmu$.
	\end{proof}
\end{proposition}
%%%%%%%%%%%%%%%%%%%%%%%%%%%%%%%%%%%%%%%%%%%%%%%%%%%%%%%%%%%%%%%%%%%%%%%%%%%%%%
%%%%%%%%%%%%%%%%%%%%%%%%%%%%%%%%%%%%%%%%%%%%%%%%%%%%%%%%%%%%%%%%%%%%%%%%%%%%%%
%%%%%%%%%%%%%%%%%%%%%%%%%%%%%%%%%%%%%%%%%%%%%%%%%%%%%%%%%%%%%%%%%%%%%%%%%%%%%%
%%%%%%%%%%%%%%%%%%%%%%%%%%%%%%%%%%%%%%%%%%%%%%%%%%%%%%%%%%%%%%%%%%%%%%%%%%%%%%
%%%%%%%%%%%%%%%%%%%%%%%%%%%%%%%%%%%%%%%%%%%%%%%%%%%%%%%%%%%%%%%%%%%%%%%%%%%%%%

\section{More $s$-pair state transfer}
\label{Section:Vto2}

In this section, we use existing vertex state transfer and $s$-pair state transfer to build more examples of $s$-pair state transfer.

%%%%%%%%%%%%%%%%%%%%%%%%%%%%%%%%%%%%%%%%%%%%%%%%%%%%%%%%%%%%%%%%%%%%%%%%%%%%%%
%%%%%%%%%%%%%%%%%%%%%%%%%%%%%%%%%%%%%%%%%%%%%%%%%%%%%%%%%%%%%%%%%%%%%%%%%%%%%%

\subsection{Fractional revival}
\label{Subsection:FR}

{\sl Fractional revival} occurs from vertex $a$ to vertex $b$ at time $\tau$ if
\begin{equation}
\label{Eqn:FR}
U(\tau) \be_a = \eta \be_a + \varpi \be_b
\end{equation}
for some complex numbers $\eta$ and $\varpi \neq 0$ satisfying $\vert\eta\vert^2+\vert \varpi \vert^2=1$.
%If $\varpi=0$ then the quantum walk is periodic at $a$ at time $\tau$.
(Vertex) perfect state transfer occurs between $a$ and $b$ if $\eta=0$, so it is a special case of fractional revival. For more information about fractional revival, we refer the reader to \cite{Chan2019}.

\begin{proposition}
	\label{Prop:P2ST-PST}
	Let $a$ and $b$ be distinct vertices in $X$.
	Let $s\in\R\backslash\{-1,0, 1\}$.  Perfect $s$-pair state transfer occurs between $\bu=\be_a+s \be_b$ and $\bmu=\be_b+ s\be_a$  if and only if
	(vertex) perfect state transfer occurs between $a$ and $b$ at the same time.
	\begin{proof}
		We first show that $\bu$ and $\bmu$ are strongly cospectral if and only if $\be_a$ and $\be_b$ are strongly cospectral.
		The states $\bu$ and $\bmu$ are strongly cospectral if and only if there exists $\sigma_\lambda \in \{\pm1\}$, for $\lambda \in \supp_\bu$,
		such that
		\begin{equation*}
		E_\lambda \bu = \sigma_\lambda E_\lambda \bmu,
		\end{equation*}
		which is equivalent to
		\begin{equation*}
		(1-s\sigma_\lambda) E_\lambda \be_a = (1-s\sigma_\lambda) \sigma_\lambda E_\lambda \be_b.
		\end{equation*}
		As $s\neq \pm 1$, the above equations hold if and only if $\be_a$ and $\be_b$ are strongly cospectral and
		\begin{equation*}
		\supp_{\be_a, \be_b}^+ = \supp_{\bu, \bmu}^+ 
		\quad\text{and}\quad 
		\supp_{\be_a, \be_b}^- = \supp_{\bu, \bmu}^-.
		\end{equation*}
		These equalities imply that $\bu$ and $\bmu$ satisfy the conditions in Theorem~\ref{Thm:RealPST} if and only if $\be_a$ and $\be_b$ satisfy 
		the conditions in Theorem~\ref{Thm:RealPST}.
	\end{proof}
\end{proposition}
In Theorem~\ref{Thm:CnPST}, we see that $C_8$ has perfect $s$-pair state transfer with $s=\pm 1$ but $C_8$ does not admit (vertex) perfect state transfer.

\begin{proposition}
	\label{Prop:PSTPeriodic}
	Let the entries of $M$ be  algebraic integers.
	Suppose (vertex) perfect state transfer occurs between $a$ and $\alpha$ at time $\tau$ and $v$ is periodic at $\tau$.
	Then perfect $s$-pair state transfer occurs between $\be_a+s\be_v$ to $\be_\alpha+s\be_v$ at time $\tau$, for $s\in \Q\backslash\{0\}$, if and only if 
	there exist $\lambda \in \supp_{\be_a,\be_\alpha}^+$ and $\lambda'\in \supp_{\be_v}$ such that
	\begin{equation*}
	\lambda \tau = \lambda' \tau \pmod{2\pi}.
	\end{equation*}
	\begin{proof}
		Let $\lambda \in \supp_{\be_a,\be_\alpha}^+$ and $\lambda' \in \supp_{\be_v}$.  By Proposition~\ref{Prop:Rat}, Theorems~\ref{Thm:RealPST} and \ref{Thm:Quadratic},
		we have
		\begin{equation*}
		U(\tau) \be_a = e^{-\ii \tau \lambda}\be_\alpha 
		\quad \text{and}\quad 
		U(\tau) \be_v = e^{-\ii \tau \lambda'}\be_v.
		\end{equation*}
		Then
		\begin{equation*}
		U(\tau)\left(\be_a + s \be_v\right) = e^{-\ii \tau \lambda}\be_\alpha +s e^{-\ii \tau \lambda'}\be_v = \eta\left(\be_\alpha +s \be_v\right)
		\end{equation*}
		if and only if $e^{-\ii \tau \lambda}=e^{-\ii \tau \lambda'}=\eta$.
	\end{proof}
\end{proposition}
Suppose the Hamiltonian $M$ is non-negative and irreducible, for example $M=A$ or $Q$. 
Then the condition on the eigenvalue supports in Proposition~\ref{Prop:PSTPeriodic} always holds with $\lambda=\lambda'$
being the Perron-Frobenius-eigenvalue of $M$.

\begin{example}
	\label{Ex:P3}
	Consider the path $P_3$ on vertex set $\{1,2,3\}$ where the vertex $2$ has degree two.
	At time $\tau=\frac{\pi}{\sqrt{2}}$,
	\begin{equation*}
	U_A(\tau)=
	-\begin{pmatrix} 0&0&1\\0&1&0\\1&0&0\end{pmatrix}
	\end{equation*}
	For any positive integer $m$, the Cartesian power $P_3^{\Box m}$ has transition matrix $U_A(t)^{\otimes m}$.
	Thus $P_3^{\Box m}$ has perfect state transfer from $a=(1,1,\ldots, 1)$ to $\alpha=(3,3,\ldots,3)$ at time $\tau$ and is periodic at $v=(2,2,\ldots,2)$ at the
	same time. As $m\sqrt{2} \in \supp_{\be_a,\be_\alpha}^+ \cap \supp_{\be_v}$, Proposition~\ref{Prop:PSTPeriodic} implies that perfect $s$-pair state transfer
	occurs from $(\be_a+s\be_v)$ to $(\be_\alpha+s\be_v)$ at time $\tau$ for all $s\in \R\backslash\{0\}$.
\end{example}

Suppose that fractional revival occurs from $a$ to $b$, and both $\eta$ and $\varpi$ in Equation~(\ref{Eqn:FR}) are non-zero.   As $U(\tau)$ is unitary, we have
\begin{equation}
\label{Eqn:P2TFR}
U(\tau) \be_b = \varpi \be_a - \varpi  \overline{\left(\frac{\eta}{\varpi}\right)} \be_b,
\end{equation}
and
\begin{equation*}
U(\tau) \left(\be_a + s \be_b\right)=(\eta + s\varpi) \be_a + \varpi   \left(1 -s \overline{\left(\frac{\eta}{\varpi}\right)} \right)\be_b.
\end{equation*}
The $s$-pair state on the right-hand side is a scalar multiple of $(\be_a+s\be_b)$ if and only if $s$ is a root
of the quadratic
\begin{equation*}
x^2+\left(\frac{\eta}{\varpi}+\overline{\left(\frac{\eta}{\varpi}\right)}\right)x-1=0.
\end{equation*}
If $s$ is a root of the above quadratic, then the $s$-pair state $(\be_a+s\be_b)$ is periodic at time $\tau$.
Otherwise, perfect $s$-pair state transfer occurs from $\be_a+s\be_b$ to the state on the right-hand side of Equation~(\ref{Eqn:P2TFR}),
but this state might not be a real state.

\begin{example}
	\label{Ex:FR}
	Let $X:=X(12,8,24)$ denote the graph obtained from identifying 8 leaves from $K_{1,20}$ and $K_{1,32}$.
	Let $a$ and $b$ denote the vertices of degree $20$ and $32$, respectively, see Figure~\ref{Fig:FR}.
	Let $A$ be the adjacency matrix of $X$.
	Fractional revival from $\be_a$ to $\be_b$ occurs at time $\tau:=\frac{\pi}{2}$ \cite{Zhang2024}, and
	\begin{equation*}
	U_A(\tau) \be_a = \frac{3}{5} \be_a - \frac{4}{5} \be_b
	\quad \text{and}\quad
	U_A(\tau)  \be_b = -\frac{4}{5} \be_a - \frac{3}{5} \be_b.
	\end{equation*}
	Then
	\begin{equation*}
	U_A(\tau)  (\be_a+s \be_b) = \frac{(3-4s)}{5}\be_a+\frac{(-4-3s)}{5}\be_b.
	\end{equation*}
	When $s\in \left\{2, -\frac{1}{2}\right\}$, the state $\be_a+s \be_b$ is periodic.
	
	\begin{figure}[h!]
		\begin{center}
			\includegraphics[scale=0.8]{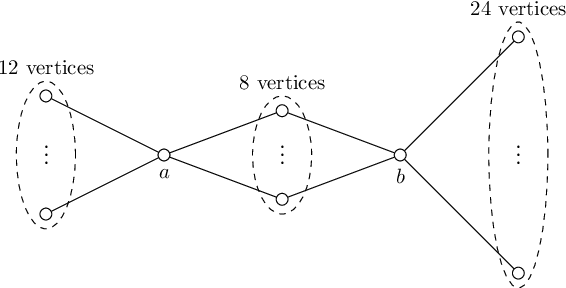}
			\caption{$X(12,8,24)$} \label{Fig:FR}
		\end{center}
	\end{figure}
	
	Let $Y:=X \Box K_2$ and $B$ be the adjacency matrix of $Y$.  Since the transition matrix of $K_2$ with adjacency Hamiltonian at time $\tau$ is 
	\begin{equation*}
	\begin{bmatrix}0 & -\ii \\ -\ii & 0\end{bmatrix},
	\end{equation*}
	we get
	\begin{equation}
	\label{Eqn:Star}
	U_{B}(\tau) \left(\be_{(a,0)} + s \be_{(b,0)} \right)=  \frac{(-3+4s)\ii}{5}\be_{(a,1)}+\frac{(4+3s)\ii}{5}\be_{(b,1)}.
	\end{equation}
	When $s\in \left\{2, -\frac{1}{2}\right\}$, perfect $s$-pair state transfer occurs in $Y$ from the $s$-pair state $\left(\be_{(a,0)} + s \be_{(b,0)}\right)$ to  $\left(\be_{(a,1)} + s \be_{(b,1)}\right)$ at time $\frac{\pi}{2}$.
\end{example}

%%%%%%%%%%%%%%%%%%%%%%%%%%%%%%%%%%%%%%%%%%%%%%%%%%%%%%%%%%%%%%%%%%%%%%%%%%%%%%
%%%%%%%%%%%%%%%%%%%%%%%%%%%%%%%%%%%%%%%%%%%%%%%%%%%%%%%%%%%%%%%%%%%%%%%%%%%%%%

\subsection{Quotient}
\label{Subsection:Quotient}

Let $M$ be a real and symmetric Hamiltonian associated with a graph $X$ on $n$ vertices.
Let $P$ be an $n\times m$ real matrix satisfying
\begin{equation}
\label{Eqn:Quo}
P^TP=I 
\quad \text{and}\quad
MP=PB,
\end{equation}
for some $m\times m$ real symmetric matrix $B$ called a quotient matrix. Then the column space of $P$ is $M$-invariant
and 
\begin{equation}
\label{Eqn:Comm}
\left(PP^T\right)M = M\left(PP^T\right).
\end{equation}
It follows from the conditions in (\ref{Eqn:Quo}) and (\ref{Eqn:Comm}) that 
\begin{equation*}
PU_B(t) = U_M(t) P.
\end{equation*}
As a result, if (vertex) perfect state transfer with Hamiltonian $B$ occurs at time $\tau$ from $\be_h$ to $\be_\ell$,
then $U_M(\tau) \left(P\be_h\right) = \eta (P\be_\ell)$.
We have perfect $s$-pair state transfer if each of $P\be_h$ and $P\be_\ell$ has two non-zero entries.

\begin{example}
	\label{Ex:C8P}
	For (\ref{Ex:C8Pa}) and (\ref{Ex:C8Pb}), let $M$ be the adjacency matrix of $C_8$, viewed as the Cayley graph on $\Z_8$ with connection set $\{-1,1\}$. 
	\begin{enumerate}[(a)]
		\item
		\label{Ex:C8Pa}
		Let $P$ be the $8\times 3$ matrix given by
		\begin{equation*}
		P^T = 
		\begin{bmatrix}
		\frac{1}{\sqrt{2}}&0&0&0&\frac{1}{\sqrt{2}}&0&0&0\\
		0&\frac{1}{2}&0&\frac{1}{2}&0&\frac{1}{2}&0&\frac{1}{2}\\
		0&0&\frac{1}{\sqrt{2}}&0&0&0&\frac{1}{\sqrt{2}}&0
		\end{bmatrix},
		\end{equation*}
		which satisfies (\ref{Eqn:Quo}).  The quotient matrix is
		\begin{equation*}
		B = \sqrt{2} \begin{bmatrix}0&1&0\\1&0&1\\0&1&0\end{bmatrix}
		\end{equation*}
		which gives
		\begin{equation*}
		U_B\left(\frac{\pi}{2}\right)\begin{bmatrix}1\\0\\0\end{bmatrix} = -\begin{bmatrix}0\\0\\1\end{bmatrix}.
		\end{equation*}
		We conclude that 
		\begin{equation*}
		U_M\left(\frac{\pi}{2}\right) \left(\be_0+\be_4\right) = - \left(\be_2+\be_6\right).
		\end{equation*}
		Thus, perfect $s$-pair state transfer occurs in $C_8$ between $\left(\be_0+\be_4\right)$ and $\left(\be_2+\be_6\right)$ at time $\frac{\pi}{2}$.
		
		\item
		\label{Ex:C8Pb}
		Let $P$ be the $8\times 4$ matrix given by
		\begin{equation*}
		P^T = 
		\begin{bmatrix}
		\frac{1}{\sqrt{2}}&0&-\frac{1}{\sqrt{2}}&0&0&0&0&0\\
		0&0&0&\frac{1}{\sqrt{2}}&0&0&0&-\frac{1}{\sqrt{2}}\\
		0&0&0&0&\frac{1}{\sqrt{2}}&0&-\frac{1}{\sqrt{2}}&0\\
		0&\frac{1}{2}&0&-\frac{1}{2}&0&\frac{1}{2}&0&-\frac{1}{2}
		\end{bmatrix},
		\end{equation*}
		which satisfies (\ref{Eqn:Quo}).  The quotient matrix is
		\begin{equation*}
		B =\begin{bmatrix}0&-1&0&0\\-1&0&1&0\\0&1&0&0\\0&0&0&0\end{bmatrix}
		\end{equation*}
		which gives
		\begin{equation*}
		U_B\left(\frac{\pi}{\sqrt{2}}\right)\begin{bmatrix}1\\0\\0\\0\end{bmatrix} = \begin{bmatrix}0\\0\\1\\0\end{bmatrix}.
		\end{equation*}
		We conclude that 
		\begin{equation*}
		U_M\left(\frac{\pi}{\sqrt{2}}\right) \left(\be_0-\be_2\right) =  \left(\be_4-\be_6\right).
		\end{equation*}
		Thus, perfect $s$-pair state transfer occurs in $C_8$ between $\left(\be_0-\be_2\right)$ and $\left(\be_4-\be_6\right)$ at time $\frac{\pi}{\sqrt{2}}$.
		
		\item
		\label{Ex:C8Pc}
		Let $M$ be the adjacency matrix of the graph $X(m)$, $m\geq 1$, given in Figure~\ref{Fig:Quo}. 
		Let $P$ be a normalized characteristic matrix corresponding to the equitable partition indicated in the figure.  (Please see \cite{Coutinho2021} for more information about equitable partitions.)
		Then the quotient matrix $B$ is the $\sqrt{m} A(C_8)$, where $A(C_8)$ is the adjacency matrix of $C_8$.
		It follows from the above examples that $X(m)$ has perfect $s$-pair state transfer between $\left(\be_0+\be_4\right)$ and $\left(\be_2+\be_6\right)$ at time $\frac{\pi}{2\sqrt{m}}$,
		and between $\left(\be_0-\be_2\right)$ and $\left(\be_4-\be_6\right)$ at time $\frac{\pi}{\sqrt{2m}}$.
		\begin{figure}[h!]
			\begin{center}
				\includegraphics[scale=0.8]{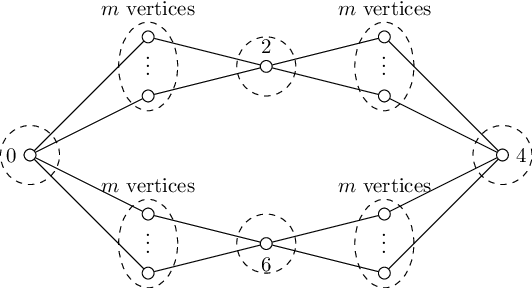}
				\caption{$X(m)$ with an equitable partition} \label{Fig:Quo}
			\end{center}
		\end{figure}
		
	\end{enumerate}
\end{example}

%%%%%%%%%%%%%%%%%%%%%%%%%%%%%%%%%%%%%%%%%%%%%%%%%%%%%%%%%%%%%%%%%%%%%%%%%%%%%%
%%%%%%%%%%%%%%%%%%%%%%%%%%%%%%%%%%%%%%%%%%%%%%%%%%%%%%%%%%%%%%%%%%%%%%%%%%%%%%

\subsection{Transitivity}
\label{Subsection:Transitivity}

We extend Theorem~6.1 of \cite{Chen2020PairST} to $s$-pair states.  To simplify the proof, we represent a real $s$-pair state $\bu=\be_a+s\be_b$ using its density matrix
\begin{equation*}
D_{\bu} = \frac{1}{(1+s^2)} \bu \bu^T.
\end{equation*}
\ignore{
	Perfect $s$-pair state transfer from $\bu$ to $\bmu$ occurs at time $\tau$ if and only if
	\begin{equation*}
	U(\tau) D_{\bu} U(-\tau) = D_{\bmu}.
	\end{equation*}
}

\begin{theorem}
	\label{Thm:Transitive}
	Suppose perfect $s$-pair state transfer occurs in graph $X$ from $\bu=\be_a + r \be_b$ to $\bmu=\be_{\alpha} + r \be_{\beta}$ at time $\tau$, and from $\bv = \be_b + s \be_c$ to $\bnu = \be_{\beta} + s \be_{\gamma}$ at time $\tau$.
	Then perfect $s$-pair state transfer occurs from $\bw=\be_a - rs \be_c$ to $\bomega=\be_\alpha - rs \be_\gamma$ at time $\tau$.
	
	\begin{proof}
		First observe that 
		\begin{equation*}
		D_\bw= \frac{1+r^2}{(1+r^2s^2)} D_\bu +  \frac{r^2(1+s^2)}{(1+r^2s^2)}D_\bv -   \frac{(1+r^2)(1+s^2)}{(1+r^2s^2)}\left(D_\bu D_\bv + D_\bv D_\bu\right),
		\end{equation*}
		and similarly,
		\begin{equation*}
		D_\bomega= \frac{1+r^2}{(1+r^2s^2)}D_\bmu +  \frac{r^2(1+s^2)}{(1+r^2s^2)}D_\bnu -   \frac{(1+r^2)(1+s^2)}{(1+r^2s^2)}\left(D_\bmu D_\bnu + D_\bnu D_\bmu \right).
		\end{equation*}
		
		The result follows from
		$U(\tau) D_\bu U(-\tau)=D_\bmu$, $U(\tau) D_\bv U(-\tau)=D_\bnu$, 
		\begin{equation*}
		U(\tau) D_\bu D_\bv U(-\tau) = \left[U(\tau) D_\bu U(-\tau)\right] \left[ U(\tau) D_\bv U(-\tau) \right]= D_\bmu D_\bnu,
		\end{equation*}
		and  $U(\tau) D_\bv D_\bu U(-\tau) = D_\bnu D_\bmu$.
	\end{proof}
\end{theorem}

\ignore{
	\begin{proposition}
		\label{Prop:Aut}
		Let the Hamiltonian of $X$ be $M \in \{A, L, Q\}$, and $\Psi$ be a graph automorphism of $X$.
		Suppose perfect $s$-pair state transfer occurs from $\bu= \be_a + s\be_b$ to $\bmu= \be_\alpha + s\be_\beta$ at time $\tau$.
		Then perfect $s$-pair state transfer occurs from $\left(\be_{\Psi(a)} + s\be_{\Psi(b)}\right)$ to $\left(\be_{\Psi(\alpha)} + s\be_{\Psi(\beta)}\right)$ at time $\tau$.
		\begin{proof}
			Let $P$ be the permutation matrix of $\Psi$.  Then  $P\be_c = \be_{\Psi(c)}$ for any vertex $c$ in $X$.
			Since $P$ commutes with $U(\tau)$,
			multiplying $P$ to both sides of $U(\tau) \bu = \eta \bmu$ yields
			\begin{equation*}
			U(\tau) \left(\be_{\Psi(a)} + s\be_{\Psi(b)}\right) = \eta \left(\be_{\Psi(\alpha)} + s\be_{\Psi(\beta)}\right).
			\end{equation*}
		\end{proof}
	\end{proposition}
}

%%%%%%%%%%%%%%%%%%%%%%%%%%%%%%%%%%%%%%%%%%%%%%%%%%%%%%%%%%%%%%%%%%%%%%%%%%%%%%
%%%%%%%%%%%%%%%%%%%%%%%%%%%%%%%%%%%%%%%%%%%%%%%%%%%%%%%%%%%%%%%%%%%%%%%%%%%%%%
%%%%%%%%%%%%%%%%%%%%%%%%%%%%%%%%%%%%%%%%%%%%%%%%%%%%%%%%%%%%%%%%%%%%%%%%%%%%%%
%%%%%%%%%%%%%%%%%%%%%%%%%%%%%%%%%%%%%%%%%%%%%%%%%%%%%%%%%%%%%%%%%%%%%%%%%%%%%%
%%%%%%%%%%%%%%%%%%%%%%%%%%%%%%%%%%%%%%%%%%%%%%%%%%%%%%%%%%%%%%%%%%%%%%%%%%%%%%

\section{Special classes}
\label{Section:Special}

In this section, we  determine all incidences of perfect $s$-pair state transfers in three families of graphs.     Since the graphs considered in this section are all regular,  
the classifications are the same for any Hamiltonian in $\{A, L, Q\}$.  Without loss of generality, we work with $U_A(t)$ in the following subsections.

%%%%%%%%%%%%%%%%%%%%%%%%%%%%%%%%%%%%%%%%%%%%%%%%%%%%%%%%%%%%%%%%%%%%%%%%%%%%%%
%%%%%%%%%%%%%%%%%%%%%%%%%%%%%%%%%%%%%%%%%%%%%%%%%%%%%%%%%%%%%%%%%%%%%%%%%%%%%%

\subsection{Complete graphs}
\label{Subsection:Kn}
Let $J_n$ denote the $n\times n$ matrix of all ones and $\one_n$ denote the vector of all ones of length $n$.
The transition matrix of $K_n$ is
\begin{equation*}
U_A(t) = \frac{e^{-\ii t (n-1)}}{n} J_n + e^{\ii t}\left(I - \frac{1}{n}J_n\right).
\end{equation*}
For any distinct vertices $a$ and $b$, $\be_a-\be_b$ is a fixed state.
Without loss of generality, let $\bu=\be_1 + s \be_2$, for some $s\neq -1$.
Then 
\begin{equation*}
U_A(t) \bu = \left(e^{-\ii t (n-1)} - e^{\ii t}\right)\frac{(1+s)}{n} \one_n + e^{\ii t}\bu.
\end{equation*}
We conclude that, for $n\geq 3$, $\bu$ is periodic with minimum period $\frac{2\pi}{n}$.

For $K_2$,  $\be_1 \pm \be_2$ are fixed states.  If $s\neq \pm 1$ then $\be_1 + s \be_2$ is periodic with minimum period $\pi$. Meanwhile, for $n\geq 2$, $K_n$ does not have perfect $s$-pair state transfer between distinct $s$-pair states.

%%%%%%%%%%%%%%%%%%%%%%%%%%%%%%%%%%%%%%%%%%%%%%%%%%%%%%%%%%%%%%%%%%%%%%%%%%%%%%
%%%%%%%%%%%%%%%%%%%%%%%%%%%%%%%%%%%%%%%%%%%%%%%%%%%%%%%%%%%%%%%%%%%%%%%%%%%%%%

\subsection{Cycles}
\label{Subsection:Cn}

We have shown in the previous section that $K_3$ does not have perfect
$s$-pair state transfer, now we consider cycles of order $n\geq 4$.

The cycle on $n$ vertices, $C_n$ is a Cayley graph on $\Z_n$ with connection set $\{-1, 1\}$.  For $j=0, \ldots, \lfloor \frac{n}{2}\rfloor$, let $\lambda_j = 2\cos \frac{2\pi j}{n}$ be the $j$-th eigenvalues of $A$, the adjacency matrix of $C_n$.
Let $E_{\lambda_j}$ be matrix of the orthogonal projection onto the $\lambda_j$-eigenspace.
Then the $(a,b)$-entry of $E_{\lambda_j}$ is
\begin{equation}
\label{EqnCn}
\begin{cases}
\frac{1}{n} & \text{if $j=0$,}\\[3pt]
\frac{2}{n}\cos \frac{2\pi j (a-b)}{n} & \text{if $1\leq j < \frac{n}{2}$,}\\[3pt]
\frac{(-1)^{a+b}}{n} & \text{if $n$ is even and $j=\frac{n}{2}$,}
\end{cases}
\end{equation}
for $a, b \in \Z_n$ and for $j=0,\ldots, \lfloor \frac{n}{2}\rfloor$.%, \cite[Section 15.10]{Coutinho2021}. 

Without loss of generality, we let the initial state be $\bu=\be_0 + s \be_b$, for some $b\in \{1, 2, \ldots, \lfloor \frac{n}{2}\rfloor\}$ and $s\in \R\backslash\{0\}$.
\ignore{
	For $c\in \Z_n$, the $c$-th entry of $E_j (\be_0 + s \be_b)$ is
	\begin{equation*}
	\be_c^T E_j (\be_0 + s \be_b) = 
	\begin{cases}
	\frac{1+s}{n} & \text{if $j=0$,}\\[5pt]
	\frac{2}{n}\left( \cos \frac{2\pi jc}{n}  + s  \cos \frac{2\pi j(c-b)}{n} \right)& \text{if $1\leq j < \frac{n}{2}$,}\\[5pt]
	\frac{(-1)^c+s(-1)^{c+b}}{n}  & \text{if $n$ is even and $j=\frac{n}{2}$.}
	\end{cases}
	\end{equation*}
}
Then $\lambda_j \in \supp_\bu$ if and only if $E_{\lambda_j} (\be_0 + s \be_b) \neq \zero$.

We first classify periodic $s$-pair states $\bu=\be_0+s\be_b$ in $C_n$, for $s=1$, $s=-1$ and $s\neq \pm 1$.
The following result \cite{lehmer1933trig} is useful in our proofs.

\begin{lemma}
	\label{Lem:Cos}
	If $\cos \frac{2\pi}{n} \in \Q$, then $n \in \{1,2,3,4,6\}$.  If the minimal polynomial of $\cos \frac{2\pi}{n}$ has degree  two, then $n\in \{5, 8, 10, 12\}$.
\end{lemma}

\begin{lemma}
	\label{Lem:CnPeriodic+}
	The $s$-pair state $\bu=\be_0+\be_b$, for $1\leq b \leq \frac{n}{2}$, is periodic in $C_n$, $n\geq 4$, with minimum period $\tau$ if and only if the triple $(n, \be_b, \tau)$ belongs to
	\begin{equation*}
	\left\{\left(4,\be_1,\pi \right), \left(4,\be_2, \frac{\pi}{2}\right), \left(6,\be_1,2\pi\right), \left(6, \be_2,2\pi\right), \left(6,\be_3, \frac{2\pi}{3}\right), \left(8,\be_4,\pi\right), \left(12, \be_6, 2\pi\right)\right\}.
	\end{equation*}
	\ignore{
		\begin{enumerate}[i.]
			\item
			$n=4$, $\bu=\be_0+\be_1$ is periodic at time $\pi$.
			\item
			$n=4$, $\bu=\be_0+\be_2$ is periodic at time $\frac{\pi}{2}$.
			\item
			$n=6$, $\bu=\be_0+\be_1$ is periodic at time $2\pi$.
			\item
			$n=6$, $\bu=\be_0+\be_2$ is periodic at time $2\pi$.
			\item
			$n=6$, $\bu=\be_0+\be_3$ is periodic at time $\frac{2\pi}{3}$.
			\item
			$n=8$, $\bu=\be_0+\be_4$ is periodic at time $\pi$.
			\item
			$n=12$, $\bu=\be_0+\be_6$ is periodic at time $2\pi$.
		\end{enumerate}
	}
	\begin{proof}
		Suppose $\bu$ is periodic.
		From (\ref{EqnCn}), for $1 \leq j < \frac{n}{2}$, 
		\begin{equation*}
		\be_0^T E_{\lambda_j} \bu =
		\frac{2}{n} \left(\cos \frac{2\pi bj}{n} + 1\right).
		\end{equation*}
		Hence $\lambda_j \in \supp_\bu$ when $\frac{bj}{n}-\frac{1}{2} \not\in \Z$.

		If $b\neq \frac{n}{2}$ then $\lambda_0, \lambda_1 \in \supp_\bu$.  Since $\lambda_0=2$, Proposition~\ref{Prop:Rat} and Theorem~\ref{Thm:Quadratic} imply $2\cos \frac{2\pi}{n}  \in \Z$.
		Similarly, if $b=\frac{n}{2}$ then $\lambda_0, \lambda_2 \in \supp_\bu$ and $2\cos \frac{4\pi}{n}  \in \Z$. It follows from Lemma~\ref{Lem:Cos} that $n\in\{4,6,8,12\}$.
		
		For $n\in \{4,6,8,12\}$, we computationally search for all $s$-pair states $\be_0+\be_b$ that satisfy the conditions in Theorem~\ref{Thm:Quadratic}
		to get the list above.  The minimum period of each case is computed using the expression given in Theorem~\ref{Thm:Quadratic}.
	\end{proof}
\end{lemma}

\begin{lemma}
	\label{Lem:CnPeriodic-}
	The $s$-pair state $\bu=\be_0-\be_b$, for $1\leq b \leq \frac{n}{2}$, is periodic in $C_n$ if and only if
	\begin{eqnarray*}
		(n, \be_b, \tau) &\in& \Bigg\{\left(4,\be_1,\pi \right), \left(5,\be_1,\frac{2\pi}{\sqrt{5}} \right), \left(5,\be_2, \frac{2\pi}{\sqrt{5}}\right), \left(6,\be_1,2\pi\right), \left(6, \be_2,\pi\right),\\
		&& \quad \left(6,\be_3, \frac{2\pi}{3}\right), \left(8,\be_2,\frac{2\pi}{\sqrt{2}}\right),\left(8,\be_4,\frac{\pi}{\sqrt{2}}\right), \left(12, \be_6, \frac{2\pi}{\sqrt{3}}\right)\Bigg\}.
	\end{eqnarray*}
	Moreover, $(\be_0-\be_2)$ is a fixed state of $C_4$.
	\ignore{
		\begin{enumerate}[i.]
			\item
			$n=5$, $\bu=\be_0-\be_1$ is periodic at time $\frac{2\pi}{\sqrt{5}}$.
			\item
			$n=5$, $\bu=\be_0-\be_2$ is periodic at time $\frac{2\pi}{\sqrt{5}}$.
			\item
			$n=6$, $\bu=\be_0-\be_1$ is periodic at time $2\pi$.
			\item
			$n=6$, $\bu=\be_0-\be_2$ is periodic at time $\pi$.
			\item
			$n=6$, $\bu=\be_0-\be_3$ is periodic at time $\frac{2\pi}{3}$.
			\item
			$n=8$, $\bu=\be_0-\be_2$ is periodic at time $\frac{2\pi}{\sqrt{2}}$.
			\item
			$n=8$, $\bu=\be_0-\be_4$ is periodic at time $\frac{\pi}{\sqrt{2}}$.
			\item
			$n=12$, $\bu=\be_0-\be_6$ is periodic at time $\frac{2\pi}{\sqrt{3}}$.
		\end{enumerate}
	}
	
	\begin{proof}
		Suppose $\bu$ is periodic.
		From (\ref{EqnCn}), for $1 \leq j < \frac{n}{2}$, 
		\begin{equation*}
		\be_0^T E_{\lambda_j} \bu =
		\frac{2}{n}\left(\cos \frac{2\pi bj}{n} - 1\right).
		\end{equation*}
		Hence $\lambda_j \in \supp_\bu$ if  $\frac{bj}{n} \not\in \Z$.

		If $1\leq b < \frac{n}{2}$ then  $\lambda_1, \lambda_2 \in \supp_\bu$.  If follows from Corollary~\ref{Cor:Gap} that $\vert \lambda_2 - \lambda_1\vert \geq 1$ which implies
		$n \leq 10$.
		Similarly, if $b=\frac{n}{2}$ then $\lambda_{j} \in \supp_\bu$ for odd $j$.  Proposition~\ref{Prop:Rat}, Theorem~\ref{Thm:Quadratic} and Lemma~\ref{Lem:Cos} imply 
		\begin{equation*}
		n\in \{4, 5, 6, 8, 10, 12\}.
		\end{equation*}
		%When $n=10$, $\lambda_1 \not \in \Z$ but $\lambda_5\in \Z$, so $\supp_\bu$ does not satisfy Theorem~\ref{Thm:Quadratic}.
		
		For $n\in \{4,5, 6,8,10,12\}$, we computationally search for all $s$-pair states $\be_0-\be_b$ that satisfy the conditions in Theorem~\ref{Thm:Quadratic}
		to get the list above.  The minimum period of each case is computed using the expression given in Theorem~\ref{Thm:Quadratic}.
	\end{proof}
\end{lemma}

\begin{lemma}
	\label{Lem:CnPeriodicR}
	For $s\neq \pm 1$,
	the $s$-pair state $\bu=\be_0 + s\be_b$ is periodic in $C_n$ if and only if one of the following holds.
	\begin{enumerate}[i.]
		\item
		\label{Lem:CnPeriodicR4}
		$n=4$, $\bu=\be_0+s \be_b$ is periodic at minimum time $\pi$, for $1\leq b \leq 3$ and $s\in \R \backslash \{\pm1\}$.
		\item
		\label{Lem:CnPeriodicR6}
		$n=6$, $\bu=\be_0+s\be_b$ is periodic at minimum time $2\pi$, for $1\leq b \leq 5$ and $s\in \R \backslash \{\pm1\}$.
	\end{enumerate}
	
	\begin{proof}
		It follows from Proposition~\ref{Prop:CospectralPeriodic} that $0$ and $b$ are periodic vertices in $C_n$, and $C_4$ and $C_6$
		are the only cycles with periodic vertices.
		
		Case~(\ref{Lem:CnPeriodicR4}) follows from the fact that the transition matrix for $C_4$ is equal to the identity matrix at minimum time $\pi$.
		Case~(\ref{Lem:CnPeriodicR6}) follows from the fact that the transition matrix for $C_6$ is equal to the identity matrix at minimum time $2\pi$.
	\end{proof}
\end{lemma}

\begin{theorem}
	\label{Thm:CnPST}
	Perfect $s$-pair state transfer occurs from initial state $\be_0+s \be_b$, for $1\leq b \leq \frac{n}{2}$, in $C_n$ if and only if 
	one of the following holds.
	\begin{enumerate}[i.]
		\item
		$n=4$, perfect $s$-pair state transfer from $\be_0+s \be_1$  to $\be_2+s \be_3$ at minimum time $\frac{\pi}{2}$, for $s\in \R\backslash\{0\}$.
		\item
		$n=4$, perfect $s$-pair state transfer from $\be_0+s \be_2$  to $\be_2+s \be_0$ at minimum time $\frac{\pi}{2}$, for $s\in \R\backslash\{-1,0,1\}$.
		\item
		$n=4$, perfect $s$-pair state transfer from $\be_0+ \be_2$  to $\be_1+ \be_3$ at minimum time $\frac{\pi}{4}$.
		\item
		$n=6$, perfect $s$-pair state transfer from $\be_0- \be_2$  to $\be_3- \be_5$ at minimum time $\frac{\pi}{2}$.
		\item
		$n=6$, perfect $s$-pair state transfer from $\be_0+2\be_2$  to $\be_0+2 \be_4$ at minimum time $\pi$.
		\item
		$n=6$, perfect $s$-pair state transfer from $\be_0+\frac{1}{2}\be_2$  to $\be_4+\frac{1}{2} \be_2$ at minimum time $\pi$.
		
		\item
		$n=8$, perfect $s$-pair state transfer from $\be_0- \be_2$  to $\be_4 - \be_6$ at minimum time $\frac{\pi}{\sqrt{2}}$.
		\item
		$n=8$, perfect $s$-pair state transfer from $\be_0+ \be_4$  to $\be_2+ \be_6$ at minimum time $\frac{\pi}{2}$.
	\end{enumerate}
	
	\begin{proof}
		By Theorem~\ref{Thm:RST} (\ref{Thm:RST3}), it is sufficient to check if perfect $s$-pair state transfer occurs at half of the minimum periods list in Lemmas~\ref{Lem:CnPeriodic+}, \ref{Lem:CnPeriodic-} and \ref{Lem:CnPeriodicR}.
	\end{proof}
\end{theorem}

%%%%%%%%%%%%%%%%%%%%%%%%%%%%%%%%%%%%%%%%%%%%%%%%%%%%%%%%%%%%%%%%%%%%%%%%%%%%%%
%%%%%%%%%%%%%%%%%%%%%%%%%%%%%%%%%%%%%%%%%%%%%%%%%%%%%%%%%%%%%%%%%%%%%%%%%%%%%%
\subsection{Antipodal distance-regular graphs}
\label{Subsection:DRG}
In this section, we determine all occurrences of perfect $s$-pair state transfers in  the distance-regular graphs that have (vertex) perfect state transfer.
See \cite{BCN1989} for a background of distance-regular graphs.

Coutinho et al.  \cite{Coutinho2015} proved that if a distance regular graph $X$ admits (vertex) perfect state transfer and $d$ is the diameter of $X$,
then every vertex in $X$ is at distance $d$ from a unique vertex.  
We call such a graph an {\sl antipodal distance-regular graph with class size two}, and
we say two vertices are {\sl antipodal} in $X$ if they are at distance $d$ from each other.

Suppose $X$ is an antipodal distance-regular graph of class size $2$ and it has diameter $d$.  For $j=0,\ldots, d$, let $A_j$ be the $j$-th distance matrix of $X$.
Then $A_0=I$, $A_d$ is the  anti-diagonal permutation matrix (up to permutation of vertices), and 
\begin{equation*}
A_j A_d = A_{d-j}, \quad \text{for $j=0,\ldots, d$,}
\end{equation*}
Let $k_j$ be the column sum of $A_j$, for $j=0,\ldots, d$.  Then $k_0=k_d=1$ and
the sequence $k_0, k_1,\ldots, k_{d-1}, k_d$ is unimodal, see Theorem 5.1.1~(i) of \cite{BCN1989}.
Further, if $X$ is not a cycle then  Theorem 5.1.1~(ii) of \cite{BCN1989} implies that $k_j \geq 3$, for $j=1,\ldots, d-1$.

Since $\{A_0, A_1, \ldots, A_d\}$ is a basis of the adjacency algebra of $X$ over $\C$, 
\begin{equation*}
U(t):=U_{A_1}(t) \in \spn \{A_0, A_1, \ldots, A_d\}.
\end{equation*}
If $X$ admits perfect state transfer between two vertices at time $\tau$ then 
\begin{equation}
\label{Eqn:A_d}
U(\tau) = \eta A_d, 
\end{equation}
for some phase factor $\eta$.

\begin{theorem}
	\label{Thm:DRG}
	Let $X$ be a distance-regular graph that is not a cycle.
	If $X$  admits (vertex) perfect state transfer,
	then $X$ has perfect $s$-pair state transfer from $\bu = \be_a + s \be_b$ to $\bmu = \be_\alpha + s \be_\beta$
	if and only if one of the following holds.
	\begin{enumerate}[(i)]
		\item
		%$\dist(a,b)<d$, $\dist(a,\alpha)=d$, $\dist(b,\beta)=d$, for $s \in \R \backslash \{0\}$.
		$\{a, b\}$ is not an antipodal pair, $\{a,\alpha\}$ and $\{b, \beta\}$ are antipodal pairs, for $s \in \R \backslash \{0\}$.
		\item
		$\{a, b\}$ is an antipodal pair, $\alpha=b$ and $\beta=a$, for $s \in \R \backslash \{-1,0,1\}$.
	\end{enumerate}
	
	\begin{proof}
		Let $d$ be the diameter of $X$, and let $\tau$ be the minimum (vertex) perfect state transfer time.
		
		Suppose $\dist(a,b)<d$.  Let $\alpha$ and $\beta$ be vertices in $X$ that are antipodal to $a$ and $b$, respectively.
		It follows from Equation~(\ref{Eqn:A_d}) that
		\begin{equation*}
		U\left(\tau \right) (\be_a + s \be_b) = \eta (\be_\alpha + s \be_\beta).
		\end{equation*}
		Since both $D_\bu$ and $D_{\bmu}$ are real, it follows from Theorem~\ref{Thm:RST}~(\ref{Thm:RST4}) that perfect state transfer does not occur from $\bu$ to another real state.
		
		Suppose $\dist(a,b)=d$, that is, $\be_b=A_d\be_a$.  Then Equation~(\ref{Eqn:A_d}) gives 
		\begin{equation*}
		U\left(\tau \right) (\be_a + s \be_b) = \eta (\be_b + s \be_a).
		\end{equation*}
		For $s\neq \pm 1$, $(\be_a + s \be_b)$ and $(\be_b + s \be_a)$ represent distinct real states and Theorem~\ref{Thm:RST}~(\ref{Thm:RST4}) implies perfect state transfer does not occur from $\bu$ to another real state. When $s=\pm 1$, the state $\bu = (I + s A_d)\be_a$ is periodic at $\tau$.  
		Let $\tau'$ be the minimum period of $\bu$.
		\ignore{
			and
			\begin{equation*}
			U(\tau')(I+sA_d)\be_a=\eta'(I+sA_d)\be_a,
			\end{equation*}
			for some phase factor $\eta'$, which implies
			\begin{equation*}
			U(\tau') (I+sA_d) = \eta'(I+sA_d).
			\end{equation*}
		}
		Let
		\begin{equation*}
		U\left(\frac{\tau'}{2}\right) = \sum_{j=0}^d \varphi_j A_j,
		\end{equation*}
		for some $\varphi_0,\ldots, \varphi_d \in \C$.
		Then
		\begin{equation*}
		U\left(\frac{\tau'}{2}\right) (\be_a + s \be_b) =U\left(\frac{\tau'}{2}\right) (I + s A_d) \be_a = \sum_{j=0}^d \left(\varphi_j + s \varphi_{d-j}\right)A_j\be_a.
		\end{equation*}
		If there exists $1\leq j \leq d-1$ such that $\varphi_j + s \varphi_{d-j} \neq 0$, then 
		$U\left(\frac{\tau'}{2}\right) (\be_a + s \be_b)$ has at least $k_j \geq 3$ non-zero entries because $X$ is not a cycle, so it is not an $s$-pair state.
		Otherwise, 
		\begin{eqnarray*}
			U\left(\frac{\tau'}{2}\right) \bu &=& U\left(\frac{\tau'}{2}\right) (I+s A_d) \be_a \\
			&=&\left( (\varphi_0+s\varphi_d) I + (\varphi_d + s \varphi_0)A_d\right)\be_a\\
			&=& (\varphi_0+s\varphi_d) \left(I+sA_d\right)\be_a\\
			&=&  (\varphi_0+s\varphi_d)\bu,
		\end{eqnarray*}
		which contradicts the assumption that $\tau'$ is the minimum period of $\bu$. By Theorem~\ref{Thm:RST}~(\ref{Thm:RST3}), we conclude that no perfect $s$-pair state transfer occurs from the states $(\be_a \pm \be_b)$.
	\end{proof}
\end{theorem}

Please see \cite{Coutinho2015} for a list of distance-regular graphs that have (vertex)
perfect state transfer, which includes the $n$-cubes.

Note that every even cycle is an antipodal distance-regular graph of class size two.  The cycle $C_4$ is the only cycle admitting (vertex) perfect state transfer.
From Theorem~\ref{Thm:CnPST}, $C_6$ and $C_8$ are the only other cycles that have perfect $s$-pair state transfer.

%%%%%%%%%%%%%%%%%%%%%%%%%%%%%%%%%%%%%%%%%%%%%%%%%%%%%%%%%%%%%%%%%%%%%%%%%%%%%%
%%%%%%%%%%%%%%%%%%%%%%%%%%%%%%%%%%%%%%%%%%%%%%%%%%%%%%%%%%%%%%%%%%%%%%%%%%%%%%

%%%%%%%%%%%%%%%%%%%%%%%%%%%%%%%%%%%%%%%%%%%%%%%%%%%%%%%%%%%%%%%%%%%%%%%%%%%%%%
%%%%%%%%%%%%%%%%%%%%%%%%%%%%%%%%%%%%%%%%%%%%%%%%%%%%%%%%%%%%%%%%%%%%%%%%%%%%%%
%%%%%%%%%%%%%%%%%%%%%%%%%%%%%%%%%%%%%%%%%%%%%%%%%%%%%%%%%%%%%%%%%%%%%%%%%%%%%%
%%%%%%%%%%%%%%%%%%%%%%%%%%%%%%%%%%%%%%%%%%%%%%%%%%%%%%%%%%%%%%%%%%%%%%%%%%%%%%
%%%%%%%%%%%%%%%%%%%%%%%%%%%%%%%%%%%%%%%%%%%%%%%%%%%%%%%%%%%%%%%%%%%%%%%%%%%%%%
\section{Line graphs}
\label{Section:Line}

Given a graph $X$, we use $V(X)$ and $E(X)$ to denote its vertex set and its edge set, respectively.
The {\sl line graph of $X$}, denoted by $\Line{X}$, is the graph with vertex set $E(X)$ and two vertices  (or two edges of $X$) are adjacent in $\Line{X}$ 
if and only if they are incident in $X$.

When $X$ admits perfect plus state transfer between $\be_a+\be_b$ and $\be_\alpha+\be_\beta$, for some edges $\{a,b\}$, $\{\alpha,\beta\}$ in $X$,
a natural question arises: does the line graph $\Line{X}$ admit (vertex) perfect state transfer between the corresponding vertices, and vice versa?
We address this question in this section, where we assume that $Q$ is the Hamiltonian of $X$ and the adjacency matrix, $\LineA$, is the Hamiltonian of $\Line{X}$.

Suppose $X$ has $n$ vertices and $m$ edges.  The {\sl incidence matrix of $X$} is an $n\times m$ matrix $R$ satisfying
\begin{equation*}
R_{a,\varepsilon} = 
\begin{cases}
1 & \text{if the vertex $a$ is incident to the edge $\varepsilon$ in $X$,}\\
0 & \text{otherwise.}
\end{cases}
\end{equation*}
Note that each column of $R$ is a plus state $\be_a+\be_b$, for some edge $\{a,b\} \in E(X)$.
It also follows immediately that
\begin{equation}
\label{Eqn:AQR}
Q = RR^T, \quad \LineA = R^TR - 2I,
\end{equation}
and
\begin{equation}
\label{Eqn:AandQ}
R U_{\LineA}(t) = e^{2\ii t} U_Q(t) R,
\end{equation}
for any time $t \in \R$.

To simplify notations, for an edge $\{a,b\}$ in $X$, we use $ab$ to denote the corresponding vertex in $\Line{X}$, and $\bff_{ab}$ for its vertex state.
We use $\Lsupp_{\bff_{ab}}$ to denote the support of $\bff_{ab}$ with respect to $\LineA$.

\begin{theorem}
	\label{Thm:LinetoQ}
	Let $\{a,b\}, \{\alpha, \beta\} \in E(X)$.
	If $\Line{X}$ admits (vertex) perfect state transfer between $\bff_{ab}$ and $\bff_{\alpha\beta}$, then $X$ admits perfect plus state transfer between
	$\be_a+\be_b$ and $\be_\alpha+\be_\beta$.
	\begin{proof}
		Suppose there exist $\tau >0$ and phase factor $\eta$ such that  $U_{\LineA}(\tau)\bff_{ab} = \eta \bff_{\alpha\beta}$.
		Left-multiplying $R$ to the both sides yields
		\begin{equation*}
		R U_{\LineA}(t) \bff_{ab}  = e^{2\ii t} U_Q(t) R\bff_{ab} = e^{2\ii t}U_Q(t) \left(\be_a+\be_b\right),
		\end{equation*}
		and
		\begin{equation*}\eta R\bff_{\alpha\beta} = \eta \left(\be_\alpha+\be_\beta\right).
		\end{equation*}
		Therefore, we get $U_Q(t) \left(\be_a+\be_b\right) = \left(\eta e^{-2\ii t}\right) \left(\be_\alpha+\be_\beta\right)$.
	\end{proof}
\end{theorem}

The null space of $R$ plays a significant role in whether the converse of Theorem~\ref{Thm:LinetoQ} holds.
Note that
\begin{equation*}
\text{nullity of  $R$} =
\begin{cases}
m-n+1 & \text{if $X$ is bipartite,}\\
m-n & \text{if $X$ is  non-bipartite.}
\end{cases}
\end{equation*}
\begin{proposition}
	\label{Prop:FullColR}
	Suppose $X$ is either a tree or a unicylic graph with an odd cycle, and
	$\{a,b\}, \{\alpha, \beta\} \in E(X)$.   
	Perfect plus state transfer occurs between $(\be_a + \be_b)$ and $(\be_{\alpha}+\be_{\beta})$ in $X$ if and only if 
	(vertex) perfect state transfer occurs between $\bff_{ab}$ and $\bff_{\alpha\beta}$ in $\Line{X}$.
	\begin{proof}
		Suppose $U_Q(\tau)(\be_a+\be_b) = \eta (\be_\alpha+\be_\beta)$.  Equation~(\ref{Eqn:AandQ}) gives
		\begin{equation*}
		R \left(U_{\LineA}(\tau) \bff_{ab}\right)  =R \left( e^{2\ii \tau}\eta  \bff_{\alpha\beta}\right).
		\end{equation*}
		As the columns of $R$ are linearly independent, we get
		\begin{equation*}
		U_{\LineA}(\tau) \bff_{ab} = e^{2\ii \tau}\eta  \bff_{\alpha\beta}.
		\end{equation*}
		The converse follows from Theorem~\ref{Thm:LinetoQ}.
	\end{proof}
\end{proposition}

\begin{example}
	\label{Ex:K2_4n}
	For $n\geq 1$, let $a$ and $\alpha$ be the two vertices of degree $4n$ in $K_{2,4n}$.  Let $b$ be a vertex of degree two in $K_{2,4n}$. 
	Then 
	\begin{equation*}
	U_Q\left(\frac{\pi}{2}\right) (\be_a+\be_b) = -(\be_\alpha+\be_b)
	\quad \text{and}\quad
	U_{\LineA}\left(\frac{\pi}{2}\right)\bff_{ab} = \bff_{\alpha b}.
	\end{equation*}
\end{example}

The converse of Theorem~\ref{Thm:LinetoQ} is not true in general, see Corollary~\ref{Cor:LnCube}.
We proceed to characterize when the converse of Theorem~\ref{Thm:LinetoQ} holds by 
investigating strong cospectrality between $\bff_{ab}$ and $\bff_{\alpha\beta}$ in $\Line{X}$.
Let the spectral decomposition of $Q$ be
\begin{equation*}
Q = \sum_{\lambda \in \spec(Q)} \lambda E_\lambda,
\end{equation*}
and the spectral decomposition of $\LineA$ be
\begin{equation*}
\LineA = \sum_{\theta \in \spec(\LineA)} \theta F_\theta.
\end{equation*}
It follows from  (\ref{Eqn:AQR}) that if $\lambda >0$ is an eigenvalue of $Q$ then $(\lambda -2)$ is an eigenvalue of $\LineA$,  and vice versa.

\begin{lemma}
	\label{Lem:EandF}
	For $\lambda \in \spec(Q)\backslash\{0\}$,  
	\begin{equation*}
	F_{(\lambda-2)} \bff_{ab} = \pm F_{(\lambda-2)} \bff_{\alpha\beta} 
	\end{equation*}
	if and only if
	\begin{equation*}
	E_\lambda (\be_a+\be_b) = \pm E_\lambda (\be_\alpha+\be_\beta).
	\end{equation*}
	\begin{proof}
		If $Q\bxx=\lambda\bxx$ with $\lambda >0$,
		then $\LineA R^T\bxx = (\lambda-2) R^T\bxx$ and the column space of $R^TE_\lambda R$ is the $(\lambda -2)$-eigenspace of $\LineA$.
		Further, as
		\begin{equation*}
		\left(R^TE_\lambda R\right)\left(R^TE_\lambda' R\right)=
		\begin{cases}
		\lambda \left(R^TE_\lambda R\right) & \text{if $\lambda = \lambda'$,}\\
		0 & \text{otherwise.}
		\end{cases}
		\end{equation*}
		We see that 
		\begin{equation}
		\label{Eqn:EF}
		F_{(\lambda-2)} = \lambda^{-1} \left(R^TE_\lambda R\right)
		\end{equation} 
		is the matrix of orthogonal projection onto the $(\lambda-2)$-eigenspace of $\LineA$.

		Left-multiplying both sides of Equation~(\ref{Eqn:EF}) with $R$ yields $RF_{(\lambda-2)} = E_\lambda R$.
		As 
		\begin{equation*}
		R\bff_{ab} = (\be_a+\be_b)
		\quad \text{and}\quad
		R\bff_{\alpha\beta}=(\be_\alpha+\be_\beta),
		\end{equation*} 
		$F_{(\lambda-2)} \bff_{ab} = \pm F_{(\lambda-2)} \bff_{\alpha\beta}$ implies 
		%$E_\lambda R\bff_{ab} = \pm E_\lambda R\bff_{\alpha\beta}$, which can be rewritten as 
		$E_\lambda \left(\be_a+\be_b\right)=\pm E_\lambda \left(\be_\alpha+\be_\beta\right)$.
		
		Conversely, 
		$E_\lambda (\be_a+\be_b) = \pm E_\lambda (\be_\alpha+\be_\beta)$ implies 
		$RF_{(\lambda-2)} \bff_{ab} = \pm RF_{(\lambda-2)} \bff_{\alpha\beta}$.
		Right multiplying both sides with $R^T$ gives 
		\begin{equation*}
		\left(\LineA+2I\right)F_{(\lambda-2)} \bff_{ab}  = \pm \left(\LineA+2I\right) F_{(\lambda-2)} \bff_{\alpha\beta}.
		\end{equation*}
		Since $\lambda\neq 0$, we get $F_{(\lambda-2)} \bff_{ab} = \pm F_{(\lambda-2)} \bff_{\alpha\beta}$.
	\end{proof}
\end{lemma}

Note that $0\in \spec(Q)$ if and only if $X$ is bipartite.  
Let $\bxx$ be an $0$-eigenvector of the signless Laplacian matrix of a bipartite graph $X$.  Then, for $\{a,b\}\in E(X)$, we have $\bxx_a=-\bxx_b$.  Thus the eigenvalue support of 
the plus state $\be_a+\be_b$ does not contain the eigenvalue $0$. 
We conclude that  strong cospectrality of the vertex states $\bff_{ab}$ and $\bff_{\alpha\beta}$ in $\Line{X}$ implies
strong cospectrality of the plus states $\be_a+\be_b$ and $\be_\alpha+\be_\beta$ in $X$.  
For the converse, we need an additional condition given in the lemma below.

\begin{lemma}
	\label{Lem:-2}
	Let $X$ be a graph with $n$ vertices and $m$ edges,
	where $m\geq n$ if $X$ is bipartite and $m>n$ otherwise.
	Suppose, for some $\{a,b\}, \{\alpha,\beta\}\in E(X)$, the plus states $(\be_a+\be_b)$ and $(\be_\alpha+\be_\beta)$ are strongly cospectral with respect to $Q$.
	Then $\bff_{ab}$ and $\bff_{\alpha\beta}$ are strongly cospectral with respect to $\LineA$ if and only if
	\begin{equation*}
	F_{-2} \bff_{ab} = \pm F_{-2}\bff_{\alpha\beta}.
	\end{equation*}
\end{lemma}

Note that $-2 \not \in \spec(\LineA)$ if and only if $X$ is a tree or a non-bipartite unicyclic graph (see Proposition~\ref{Prop:FullColR}).
For all other graphs, we see from  the proofs of Theorems~1 and 2 in \cite{Akbari} that  $-2 \not\in \Lambda_{\bff_{ab}}$  only if $\{a,b\}$ is a cut-edge in $X$.

\begin{lemma}
	\label{Lem:Bip}
	Suppose $X$ is a bipartite graph on $n$ vertices with $m\geq n$ edges.
	If $\bff_{ab}$ and $\bff_{\alpha\beta}$ are strongly cospectral vertices in $\Line {X}$, then $\left\{\{a,b\}, \{\alpha, \beta\}\right\}$ is an edge-cut in $X$.
	\begin{proof}
		Let $\{a,b\}$ and $\{\alpha, \beta\}$ be edges in $X$ that do not form an edge-cut of $X$.
		From the proof of Theorem~1 of \cite{Akbari}, using a spanning tree of $X\backslash \left\{\{a,b\}, \{\alpha, \beta\}\right\}$, we can construct a vector $\by$  in the null space of $R$ such that
		\begin{equation*}
		\bff_{ab}^T \by = 1\quad \text{and}\quad \bff_{\alpha\beta}^T \by = 0.
		\end{equation*}
		Now $\by$ is an $(-2)$-eigenvector of $\LineA$, which implies
		$F_{-2} \bff_{ab} \neq \pm F_{-2}\bff_{\alpha\beta}$.
		The result follows from Lemma~\ref{Lem:-2}.
	\end{proof}
\end{lemma}

%In \cite{??}, Cao (??) observed that the $3$-cube has plus state transfer but its line graph does not have vertex state transfer. {\color{red} check}
Theorem~\ref{Thm:DRG} implies that $n$-cube has plus state transfer, for $n \geq 2$.  We now rule out perfect state transfer 
in the line graph of $n$-cube, for $n\geq 3$.

\begin{corollary}
	\label{Cor:LnCube}
	For $n\geq 3$, the line graph of the $n$-cube does not admit (vertex) perfect state transfer.
	\begin{proof}
		It is well-known that the $n$-cube is a bipartite graph with edge-connectivity $n$.  This result follows immediately from Lemma~\ref{Lem:Bip}.
	\end{proof}
\end{corollary}

\begin{lemma}
	\label{Lem:nonBip}
	Suppose $X$ is a non-bipartite graph on $n$ vertices with $m > n$ edges.
	If $\bff_{ab}$ and $\bff_{\alpha\beta}$ are strongly cospectral vertices in $\Line {X}$, then the removal of the edges $\{a,b\}$ and $\{\alpha,\beta\}$ from $X$
	results in either a disconnected graph or a bipartite graph.
	\begin{proof}
		Let $\{a,b\}$ and $\{\alpha, \beta\}$ be edges in $X$ whose removal result in a connected non-bipartite subgraph of $X$.
		From the proof of Theorem~2 of \cite{Akbari}, there exists a vector $\by$  in the null space of $R$ such that
		\begin{equation*}
		\bff_{ab}^T \by = 1\quad \text{and}\quad \bff_{\alpha\beta}^T \by = 0.
		\end{equation*}
		Now $\by$ is an $(-2)$-eigenvector of $\LineA$, which implies
		$F_{-2} \bff_{ab} \neq \pm F_{-2}\bff_{\alpha\beta}$.
		The result follows from Lemma~\ref{Lem:-2}.
	\end{proof}
\end{lemma}

We are ready to characterize when $\Line{X}$ has (vertex) perfect state transfer,  based on information from the signless Laplacian matrix of $X$.
\begin{theorem}
	\label{Thm:LinePST}
	Let $\bu=\be_a+\be_b$ and $\bmu=\be_\alpha+\be_\beta$, for some edges $\{a,b\}$ and  $\{\alpha,\beta\}$ in $X$.
	The line graph of $X$ admits (vertex) perfect state transfer between $\bff_{ab}$ and $\bff_{\alpha\beta}$ if and only if the following conditions hold.
	\begin{enumerate}[(i)]
		\item
		\label{Thm:LinePST1}
		$X$ admits perfect plus state transfer between $\bu$ and $\bmu$.
		\item
		\label{Thm:LinePST2}
		If $-2 \in \Lambda_{\bff_{ab}} $ then $F_{-2}\bff_{ab} = \pm F_{-2}\bff_{\alpha\beta}$, and either the elements in $\Lambda_{\bff_{ab}} = \Lambda_{\bff_{\alpha\beta}}$
		are
		\begin{enumerate}[(i)]
			\item
			\label{Thm:LinePST2a}
			all integers, or
			\item
			\label{Thm:LinePST2b}
			%$-2  \in \Lambda_{\bff_{ab}}$ and 
			there exists a square-free integer $\Delta>1$ such that each element of $\Lambda_{\bff_{ab}}$
			is in the form $-2+\frac{d\sqrt{\Delta}}{2}$, for some $d\in \Z$.
		\end{enumerate}
		\item
		\label{Thm:LinePST3}
		Let $\theta' \in  \Lambda_{\bff_{ab}, \bff_{\alpha\beta}}^+$, and
		\begin{equation*}
		g=\gcd \left\{\frac{\theta' - \theta}{\sqrt{\Delta}}\right\}_{\theta \in \Lambda_{\bff_{ab}}}
		\end{equation*}
		(with $\Delta=1$ for Case~(a) above).
		Then $\theta \in \Lambda_{\bff_{ab}, \bff_{\alpha\beta}}^+$ if and only if $\frac{\theta' - \theta}{g\sqrt{\Delta}}$ is even.
	\end{enumerate}
	\begin{proof}
		If $-2 \not \in \Lambda_{\bff_{ab}}$ then $\spec(\LineA) = \{\theta-2 : \theta \in \supp_\bu\}$, and 
		it follows from Lemma~\ref{Lem:EandF} and Theorem~\ref{Thm:RealPST} that Condition~(\ref{Thm:LinePST1}) implies perfect state transfer occurring
		between $\bff_{ab}$ and $\bff_{\alpha\beta}$ in $\Line{X}$.
		
		Suppose $-2 \in \Lambda_{\bff_{ab}}$.  Then Condition~(\ref{Thm:LinePST1}), $F_{-2}\bff_{ab} = \pm F_{-2}\bff_{\alpha\beta}$ and Lemma~\ref{Lem:EandF}
		imply that $\bff_{ab}$ and $\bff_{\alpha\beta}$ are strongly cospectral with respect to $\LineA$. Conditions~(\ref{Thm:LinePST2}) and (\ref{Thm:LinePST3})
		are equivalent to Theorem~\ref{Thm:RealPST}~(\ref{Thm:RealPST2}) and (\ref{Thm:RealPST3}).
	\end{proof}
\end{theorem}

\begin{remark}
	If $-2 \in \Lambda_{\bff_{ab}} $ then $\Lambda_{\bff_{ab}} =  \{\theta-2 : \theta \in \supp_\bu\} \cup \{-2\}$.
	Suppose Theorem~\ref{Thm:LinePST}~(\ref{Thm:LinePST2b}) holds, let 
	\begin{equation*}
	h=\gcd \left\{\frac{\lambda}{\sqrt{\Delta}}, \frac{\lambda-\lambda'}{\sqrt{\Delta}}\right\}_{\lambda, \lambda' \in \Phi_{\bu}}
	\end{equation*}
	(with $\Delta=1$ for Case (a) above).
	Then Theorem~\ref{Thm:LinePST}~(\ref{Thm:LinePST3}) can be expressed as
	\begin{equation*}
	\Phi^\varsigma_{\bu,\bmu} = \left\{\lambda \in \Phi_{\bu}  : \frac{\lambda}{h\sqrt{\Delta}} \text{ is even}\right\}
	\end{equation*}
	if $-2 \in \Lambda^\varsigma_{\bff_{ab},\bff_{\alpha\beta}}$, for $\varsigma\in \{+,-\}$.
\end{remark}

%%%%%%%%%%%%%%%%%%%%%%%%%%%%%%%%%%%%%%%%%%%%%%%%%%%%%%%%%%%%%%%%%%%%%%%%%%%%%%
%%%%%%%%%%%%%%%%%%%%%%%%%%%%%%%%%%%%%%%%%%%%%%%%%%%%%%%%%%%%%%%%%%%%%%%%%%%%%%
%%%%%%%%%%%%%%%%%%%%%%%%%%%%%%%%%%%%%%%%%%%%%%%%%%%%%%%%%%%%%%%%%%%%%%%%%%%%%%
%%%%%%%%%%%%%%%%%%%%%%%%%%%%%%%%%%%%%%%%%%%%%%%%%%%%%%%%%%%%%%%%%%%%%%%%%%%%%%
%%%%%%%%%%%%%%%%%%%%%%%%%%%%%%%%%%%%%%%%%%%%%%%%%%%%%%%%%%%%%%%%%%%%%%%%%%%%%%
%\newpage
\section{Line graphs of Cartesian products}
\label{Section:Cart}

In Corollary~\ref{Cor:LnCube}, we ruled out (vertex) perfect state transfer in the line graph of  the $n$-cube, for $n\geq 3$.
In this section, we characterize adjacency perfect state transfer in $\Line{X}$ when $X$ is a Cartesian product of two connected graphs.

We introduce some notation for this section. 
Let $X=X_1\square X_2$ where $X_1$ and $X_2$ are connected graphs on $n_1$ and $n_2$ vertices with $m_1$ and $m_2$ edges, respectively. 
We assume $n_1, n_2 \geq 2$.
For $j=1,2$, we let $R_j$  be the $n_j\times m_j$ $01$-incidence matrices of $X_j$, and let $r_j=\rank R_j$.
From \cite{fan2022stabilizing}, we can assume that $R$ has the form
\begin{equation}
\label{Eqn:R}
R = \begin{bmatrix}
I_{n_1}\otimes R_2 \;\;&\;\; R_1\otimes I_{n_2}
\end{bmatrix},
\end{equation}
which has rank $\left(n_1r_2+r_1n_2-r_1r_2\right)$.

Let $Q_j=R_jR_j^T$ be the signless Laplacian matrix of $X_j$, for $j=1,2$.
Then the signless Laplacian matrix of $X_1\square X_2$ is
\begin{equation*}
\LineQ = I_{n_1}\otimes Q_2 + Q_1\otimes I_{n_2},
\end{equation*}
and the adjacency matrix, $\LineA$, of $\Line{X_1\square X_2}$ satisfies
\begin{equation*}
\renewcommand*{\arraystretch}{1.25}
\LineA + 2I_{(n_1m_2+n_2m_1)} =
\begin{bmatrix}
I_{n_1}\otimes \left(\LineA_2 + 2I_{n_2}\right) & R_1\otimes R_2^T\\
R_1^T \otimes R_2 & (\LineA_1+2I_{n_1}) \otimes I_{n_2}
\end{bmatrix},
\end{equation*}
where $\LineA_j$ is the adjacency matrix of $\Line{X_j}$ for $j=1,2$.

\begin{figure}[h!]
	\begin{center}
		\includegraphics[scale=0.75]{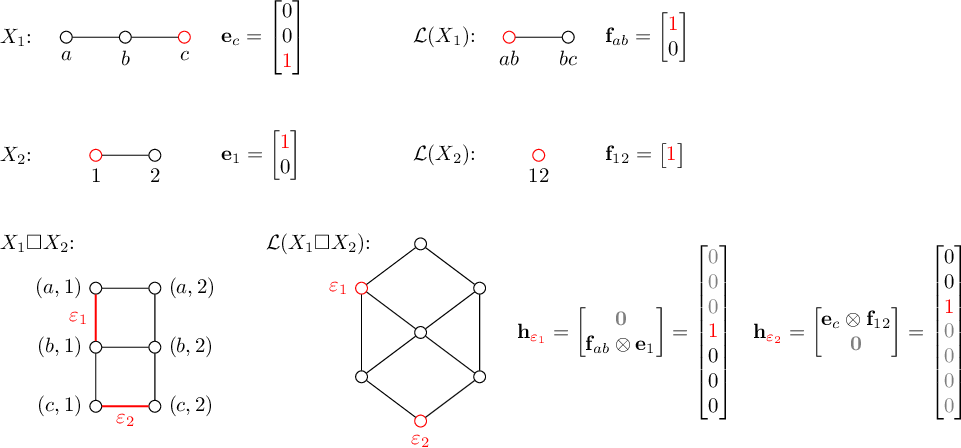}
		\caption{$\Line{P_3\square K_2}$} \label{Fig:Cart}
	\end{center}
\end{figure}

For an edge $\varepsilon$ of $X_1\square X_2$, we use $\bh_{\varepsilon}$ to denote the characteristic vector of $\varepsilon$ as a vertex state of $\Line{X_1\square X_2}$.
If $\varepsilon_1$ joins vertices $(a,\gamma)$ and $(b, \gamma)$ in $X_1\square X_2$, for some $\{a, b\} \in E(X_1)$ and $\gamma \in V(X_2)$, then 
\begin{equation*}
\bh_{\varepsilon_1} = \begin{bmatrix}\zero\\ \bff_{ab} \otimes \be_{\gamma}\end{bmatrix}.
\end{equation*}
(Recall $\{a,b\}$ is a vertex in $\Line{X_1}$ and $\bff_{ab}$ is the corresponding vertex state in $\Line{X_1}$.  Also $\be_{\gamma}$ is the vertex state for $\gamma$ in $X_2$.)
If $\varepsilon_2$ joins vertices $(c,\alpha)$ and $(c, \beta)$, for some $c \in V(X_1)$ and $\{\alpha, \beta\} \in E(X_2)$, then 
\begin{equation*}
\bh_{\varepsilon_2} = \begin{bmatrix} \be_{c}\otimes \bff_{\alpha\beta} \\ \zero \end{bmatrix}.
\end{equation*}
See Figure~\ref{Fig:Cart} for an example of these vectors when $X_1= P_3$ and $X_2=K_2$.

Suppose $Q_j \bxx_j = \lambda_j \bxx_j$, for $j=1,2$.  Then 
\begin{equation*}
\LineQ \left(\bxx_1\otimes \bxx_2\right) = \left(\lambda_1+\lambda_2\right)\left( \bxx_1\otimes \bxx_2\right)
\end{equation*} 
and
\begin{equation}
\label{Eqn:AL}
\renewcommand*{\arraystretch}{1.25}
\LineA R^T\left( \bxx_1\otimes \bxx_2\right) = \LineA \begin{bmatrix} \bxx_1\otimes R_2^T \bxx_2\\ R_1^T \bxx_1 \otimes \bxx_2\end{bmatrix} = 
\left(\lambda_1+\lambda_2-2\right) \begin{bmatrix} \bxx_1\otimes R_2^T \bxx_2\\ R_1^T \bxx_1 \otimes \bxx_2\end{bmatrix}.
\end{equation}
Hence $ \left(\lambda_1+\lambda_2-2\right)$ belongs to the eigenvalue support of $\bh_{\varepsilon_1}$ with respect to $\LineA$ only if
$\bff_{ab}^TR_1^T\bxx_1 \neq \zero$, which implies $\left(\lambda_1-2\right) \in \Lambda_{\bff_{ab}}$, the eigenvalue support of $\bff_{ab}$ with respect to $\LineA_1$.
Similarly, if $ \left(\lambda_1+\lambda_2-2\right)$ belongs to the eigenvalue support of $\bh_{\varepsilon_2}$, then $\left(\lambda_2-2\right) \in \Lambda_{\bff_{\alpha\beta}}$,
the eigenvalue support of $\bff_{\alpha\beta}$ with respect to $\LineA_2$.

For $j=1,2$, let $N_j$ be a $m_j \times (m_j-r_j)$ matrix whose column space is equal to the null space of $R_j$.
Let $N$ be a  matrix whose column space is the null space of $R$.  The matrix $N$ has $(n_1m_2+n_2m_1)$ rows and $\left(n_1(m_2-r_2)+n_2(m_1-r_1)+r_1r_2\right)$ columns.
Moreover, it follows from the proofs of Theorems~1 and 2 in \cite{Akbari} that $N_j$ has a zero row  only if the corresponding edge is a cut-edge in $X_j$.

\begin{proposition}
	\label{Prop:N}
	Let $X=X_1 \square X_2$.
	For $j=1,2$, let $\wt{R}_j$ be a matrix consisting $r_j$ linearly independent columns of $R_j$,
	and let $\wt{I}_j$ be the $(m_j \times r_j)$ matrix whose $l$-th column is $\bff_{ab}$
	if the $l$-th column of $\wt{R_j}$ is the column in $R_j$ indexed by $\{a,b\} \in E(X_j)$.
	Then the column space of
	\begin{equation}
	\renewcommand*{\arraystretch}{1.25}
	\label{Eqn:N}
	N=
	\begin{bmatrix}
	I_{n_1} \otimes N_2 & \zero & -\wt{R}_1 \otimes \wt{I}_2\\ \zero & N_1 \otimes I_{n_2} & \wt{I}_1 \otimes \wt{R}_2
	\end{bmatrix}
	\end{equation}
	is exactly the null space of $R$.  
	\begin{proof}
		For $j=1,2$, we have $R_j \wt{I}_j = \wt{R}_j$ and $R_jN_j=\zero$.  It follows immediately that $RN=\zero$.
		Since $N$ has $\big(n_1(m_2-r_2)+n_2(m_1-r_1)+r_1r_2\big)$ linearly independent columns, we conclude that the column space of $N$ is equal to 
		the null space of $R$.
	\end{proof}
\end{proposition}
In the following, we assume $N$ is the matrix given in Equation~(\ref{Eqn:N}).
Let $F_{-2}$ denote the matrix of orthogonal projection on the $(-2)$-eigenspace of 
the adjacency matrix $\LineA$ of $\Line{X_1\square X_2}$, which is equal to the null space of $R$.
It follows  that, for vertices $\varepsilon_1$ and $\varepsilon_2$ in $\Line{X_1\square X_2}$, $F_{-2}\left(\bh_{\varepsilon_1}\pm \bh_{\varepsilon_2}\right) = \zero$ if and only if $\left(\bh_{\varepsilon_1}\pm \bh_{\varepsilon_2} \right)^T N= \zero$.

\begin{lemma}
	\label{Lem:SCh}
	Let $\varepsilon_1$ and $\varepsilon_2$ be distinct vertices of $\Line{X_1\square X_2}$.  Then 
	$F_{-2}\bh_{\varepsilon_1} = \pm F_{-2}\bh_{\varepsilon_2}$
	if and only if one of the following holds.
	\begin{enumerate}[(i)]
		\item
		\label{Lem:SCh1}
		Without loss of generality, $X_2=K_2$.  Let $V(X_2)=\{1,2\}$. 
		Then
		\begin{equation*}
		\bh_{\varepsilon_1} = \begin{bmatrix}\zero\\ \bff_{a b} \otimes \be_{1}\end{bmatrix}
		\quad \text{and}\quad 
		\bh_{\varepsilon_2} = \begin{bmatrix}\zero\\ \bff_{a b} \otimes \be_{2}\end{bmatrix},
		\end{equation*}
		for some edge $\{a,b\}$ in $X_1$ such that $-2 \not \in \Lambda_{\bff_{ab}}$.
		\item
		\label{Lem:SCh2}
		\begin{equation*}
		\bh_{\varepsilon_1} = \begin{bmatrix}\be_a \otimes \bff_{\alpha \beta} \\ \zero\end{bmatrix}
		\quad \text{and}\quad 
		\bh_{\varepsilon_2} = \begin{bmatrix}\zero\\ \bff_{a b} \otimes \be_{\alpha}\end{bmatrix},
		\end{equation*}
		for some pendant vertices $a \in V(X_1)$ and $\alpha \in V(X_2)$,  and edges $\{a,b\} \in E(X_1)$ and $\{\alpha, \beta\} \in E(X_2)$
		such that $-2 \not\in\Lambda_{\bff_{ab}}$ and $-2 \not\in\Lambda_{\bff_{\alpha\beta}}$.
	\end{enumerate}
	In addition, $F_{-2}\bh_{\varepsilon_1} = F_{-2}\bh_{\varepsilon_2}\neq \zero$ in Case~(\ref{Lem:SCh1}),
	and $F_{-2}\bh_{\varepsilon_1} = -F_{-2}\bh_{\varepsilon_2}\neq \zero$ in Case~(\ref{Lem:SCh2}).
\end{lemma}

\begin{proof}
	Suppose $\left(\bh_{\varepsilon_1}\pm \bh_{\varepsilon_2} \right)^T N= \zero$.   
	There are two possible cases.
	
	Without loss of generality, for $j=1,2$, the edge $\varepsilon_j$ joins the vertices $(a_j,\gamma_j)$ to $(b_j,\gamma_j)$ in $X_1\square X_2$, 
	for some $\{a_j,b_j\} \in E(X_1)$  and  $\gamma_j \in V(X_2)$.  That is,
	\begin{equation*}
	\bh_{\varepsilon_1} = \begin{bmatrix}\zero\\ \bff_{a_1 b_1} \otimes \be_{\gamma_1}\end{bmatrix}
	\quad \text{and}\quad
	\bh_{\varepsilon_1} = \begin{bmatrix}\zero\\ \bff_{a_2 b_2} \otimes \be_{\gamma_2}\end{bmatrix}.
	\end{equation*}
	(It is possible that $\{a_1,b_1\}=\{a_2,b_2\}$ or $\gamma_1=\gamma_2$.)
	We can  then assume that the columns of $\wt{R}_1$ contain $R_1\bff_{a_1b_1}$ and $R_1\bff_{a_2b_2}$.
	Thus for $j=1,2$,
	\begin{equation*}
	\bff_{a_jb_j}^T\wt{I}_1 \neq \zero.
	\end{equation*}
	
	From Equation~(\ref{Eqn:N}), $\left(\bh_{\varepsilon_1}\pm \bh_{\varepsilon_2} \right)^T N= \zero$ implies
	\begin{equation}
	\label{Eqn:SCh1}
	\bff_{a_1b_1}^T\wt{I}_1 \otimes \be_{\gamma_1}^T\wt{R}_2 \pm
	\bff_{a_2b_2}^T\wt{I}_1 \otimes \be_{\gamma_2}^T\wt{R}_2 = \zero.
	\end{equation}
	Note that $\bff_{a_jb_j}^T\wt{I}_1$ and $\be_{\gamma_j}\wt{R}_2$ are $01$-vectors, for $j=1,2$.
	Equation~(\ref{Eqn:SCh1}) holds only for the difference of the two terms
	with $\bff^T_{a_1b_1}\wt{I}_1=\bff^T_{a_2b_2}\wt{I}_1$ and $\be_{\gamma_1}^T\wt{R}_2=\be_{\gamma_2}^T\wt{R}_2$.
	We conclude that $\{a_1,b_1\}=\{a_2,b_2\}$ in $X_1$, and $\gamma_1$ and $\gamma_2$ are incident to the same set of edges in $X_2$.
	Since $\varepsilon_1\neq \varepsilon_2$, $\gamma_1$ and $\gamma_2$ are distinct vertices in $X_2$.  As $X_2$ is simple, 
	we conclude that $X_2=K_2$.
	
	Further $\left(\bh_{\varepsilon_1}\pm \bh_{\varepsilon_2} \right)^T N= \zero$ implies that $\bff_{ab}^TN_1=0$ (unless $N_1$ has no column).
	Thus $-2 \not\in \Lambda_{\bff_{ab}}$. This proves (a).

	For the second case, we assume without loss of generality that
	\begin{equation*}
	\bh_{\varepsilon_1} = \begin{bmatrix} \be_{c}\otimes \bff_{\alpha\beta} \\ \zero \end{bmatrix}
	\quad\text{and}\quad
	\bh_{\varepsilon_2} = \begin{bmatrix}\zero\\ \bff_{ab} \otimes \be_{\gamma}\end{bmatrix}.
	\end{equation*}
	We can assume that the columns of $\wt{R}_1$ contain $R_1\bff_{ab}$ and the columns of $\wt{R}_2$ contain $R_2\bff_{\alpha\beta}$.
	From Equation~(\ref{Eqn:N}), $\left(\bh_{\varepsilon_1}\pm \bh_{\varepsilon_2} \right)^T N= \zero$ yields
	\begin{equation}
	\label{Eqn:SCh2}
	-\be_c^T\wt{R}_1\otimes \bff_{\alpha\beta}^T\wt{I}_2 \pm \bff_{ab}^T\wt{I}_1\otimes \be_{\gamma}^T\wt{R}_2 = \zero.
	\end{equation}
	Since $\be_c^T\wt{R}_1$, $\bff_{\alpha\beta}^T\wt{I}_2$, $\bff_{ab}^T\wt{I}_1$ and $\be_\gamma^T \wt{R}_2$ are $01$-vectors, Equation~(\ref{Eqn:SCh2}) 
	holds only when we add the two terms with
	\begin{equation*}
	\be_c^T\wt{R}_1 = \bff_{ab}^T\wt{I}_1
	\quad\text{and}\quad
	\bff_{\alpha\beta}^T\wt{I}_2=  \be_{\gamma}^T\wt{R}_2.
	\end{equation*}
	This implies $c\in \{a,b\}$ and $\gamma\in\{\alpha,\beta\}$.
	Without loss of generality, let $c=a$ and $\gamma=\alpha$.  That is, $\varepsilon_1$ joins vertices $(a,\alpha)$ and $(a,\beta)$  and $\varepsilon_2$ joins vertices $(a,\alpha)$ and $(b,\alpha)$  in $X_1\square X_2$.

	Suppose $X_1\square X_2$ is bipartite.  By the proof of Lemma~\ref{Lem:Bip}, 
	$F_{-2}\left(\bh_{\varepsilon_1}\pm \bh_{\varepsilon_2}\right) = \zero$ implies $\{\varepsilon_1,\varepsilon_2\}$ is an edge-cut in $X_1\square X_2$.
	
	Suppose, without loss of generality $X_1$ is non-bipartite so that $X_1\square X_2$ is not bipartite.
	But $X_1\square X_2 \backslash\{\varepsilon_1, \varepsilon_2\}$ contains an odd cycle, so
	it follows from  the proof of Lemma~\ref{Lem:nonBip} that $\{\varepsilon_1,\varepsilon_2\}$ is an edge-cut in $X_1\square X_2$.
	
	The set $\{\varepsilon_1,\varepsilon_2\}$ is an edge-cut in $X_1\square X_2$ only when $a$ is a pendant vertex in $X_1$ and $\alpha$ is a pendant vertex in $X_2$.\
	
	In addition,  $\left(\bh_{\varepsilon_1}\pm \bh_{\varepsilon_2} \right)^T N= \zero$ implies that $\bff_{ab}^TN_1=0$ (unless $N_1$ has no column) and $\bff_{\alpha\beta}^TN_2=0$ (unless $N_2$ has no column). Thus $-2 \not\in \Lambda_{\bff_{ab}}$ and $-2 \not\in \Lambda_{\bff_{\alpha\beta}}$. This proves (b).
\end{proof}

We first consider Case~(\ref{Lem:SCh1}) of Lemma~\ref{Lem:SCh}.
\begin{lemma}
	\label{Lem:SCha}
	Let $\varepsilon_1$ and $\varepsilon_2$ be vertices of $\Line{X_1\square X_2}$ satisfying Lemma~\ref{Lem:SCh}~(\ref{Lem:SCh1}).
	Then one of the following holds.
	\begin{enumerate}[(i)]
		\item
		$X_1$ is a non-bipartite unicyclic graph and $\{a,b\}$ is an edge in the odd cycle of  $X_1$, or
		\item
		$\{a,b\}$ is a cut-edge in $X_1$.
	\end{enumerate}
	\begin{proof}
		In Lemma~\ref{Lem:SCh}~(\ref{Lem:SCh1}), 
		if $N_1$ has at least one column then $\bff_{ab}^TN_1=\zero$ implies that $\{a,b\}$ is a cut-edge in $X_1$.
		If $N_1$ has no column then $X_1$ is either a tree or a unicycle graph with an odd cycle.
	\end{proof}
\end{lemma}

We use $\Psi_{\bh_{\varepsilon}}$ to denote the eigenvalue support of the
vertex state $\bh_{\varepsilon}$ with respect to $\LineA$.
For vertex states $\bh_{\varepsilon_1}$ and $\bh_{\varepsilon_2}$, let
\begin{equation*}
\Psi^{\pm}_{\bh_{\varepsilon_1},\bh_{\varepsilon_2}} = \left\{\vartheta\in \Psi_{\bh_{\varepsilon_1}} : F_\vartheta \bh_{\varepsilon_1}= \pm F_\vartheta \bh_{\varepsilon_2}\right\},
\end{equation*}
where $F_\vartheta$ is the matrix of orthogonal projection onto the $\vartheta$-eigenspace of $\LineA$.

\begin{lemma}
	\label{Lem:SChaIFF}
	Let $\varepsilon_1$ and $\varepsilon_2$ be vertices of $\Line{X_1\square X_2}$ satisfying Lemma~\ref{Lem:SCh}~(\ref{Lem:SCh1}).
	Then 
	$\bh_{\varepsilon_1}$ and $\bh_{\varepsilon_2}$ are strongly cospectral with respect to $\LineA$ if and only if
	$\vert \vartheta-\vartheta'\vert \neq 2$, for $\vartheta, \vartheta' \in \Lambda_{\bff_{ab}}$.

	Further, if $\bh_{\varepsilon_1}$ and $\bh_{\varepsilon_2}$ are strongly cospectral  then
	\begin{equation}
	\label{Eqn:SCh1b-1}
	\Psi^{-}_{\bh_{\varepsilon_1},\bh_{\varepsilon_2}} = \Lambda_{\bff_{ab}},
	\quad 
	\Psi^{+}_{\bh_{\varepsilon_1},\bh_{\varepsilon_2}} = \left\{\vartheta+2 : \vartheta \in  \Lambda_{\bff_{ab}}\right\} \cup \{-2\}.
	\end{equation}
	\begin{proof}
		Let $Q_1 \bxx = \lambda \bxx$ and
		\begin{equation*}
		\by_1 = \begin{bmatrix}1\\1\end{bmatrix}
		\quad \text{and}\quad 
		\by_2 = \begin{bmatrix}1\\-1\end{bmatrix}.
		\end{equation*}
		From Equation~(\ref{Eqn:AL}), for $s=1,2$,
		\begin{equation*}
		\renewcommand*{\arraystretch}{1.25}
		\LineA \begin{bmatrix} \bxx \otimes R_2^T\by_s\\R_1^T\bxx \otimes \by_s\end{bmatrix}
		=\vartheta_s  \begin{bmatrix} \bxx \otimes R_2^T\by_s\\R_1^T\bxx \otimes \by_s\end{bmatrix},
		\end{equation*}
		where $\vartheta_1=\lambda$ and $\vartheta_2=\lambda-2$.
		For $j,s\in\{1,2\}$,
		\begin{equation*}
		\renewcommand*{\arraystretch}{1.25}
		\bh_{\varepsilon_j}^T \begin{bmatrix} \bxx \otimes R_2^T\by_s\\R_1^T\bxx \otimes \by_s\end{bmatrix} = 
		\begin{cases}
		-\bff_{ab}^TR_1^T\bxx & \text{if $s=j=2$,}\\
		\bff_{ab}^TR_1^T\bxx & \text{otherwise.}
		\end{cases}
		\end{equation*}
		Hence $\lambda, \lambda-2 \in \Psi_{\bh_{\varepsilon_j}}$ if and only if $\lambda-2 \in \Lambda_{\bff_{ab}}$, which implies $\Psi_{\bh_{\varepsilon_1}}=\Psi_{\bh_{\varepsilon_2}}$.
		In addition,
		\begin{equation}
		\label{Eqn:SCh1b-2}
		\bh_{\varepsilon_1}^T \begin{bmatrix} \bxx \otimes R_2^T\by_s\\R_1^T\bxx \otimes \by_s\end{bmatrix}  = (-1)^{s-1} \bh_{\varepsilon_2}^T \begin{bmatrix} \bxx \otimes R_2^T\by_s\\R_1^T\bxx \otimes \by_s\end{bmatrix}.
		\end{equation}
		
		Suppose $\Lambda_{\bff_{ab}}$ contains $\vartheta$ and $\vartheta'=\vartheta+2$.  Let $\bxx$ and $\bxx'$ be an $\vartheta$-eigenvector and an $\vartheta'$-eigenvector of $Q_1$, respectively,
		such that $\bff_{ab}^TR_1^T \bxx$ and  $\bff_{ab}^TR_1^T \bxx'$ are non-zero.
		Then $\vartheta \in \Psi_{\bh_{\varepsilon_j}}$, for $j=1,2$, and
		\begin{equation*}
		\bh_{\varepsilon_1}^T \begin{bmatrix} \bxx \otimes R_2^T\by_1\\R_1^T\bxx \otimes \by_1\end{bmatrix}  =  \bh_{\varepsilon_2}^T \begin{bmatrix} \bxx \otimes R_2^T\by_1\\R_1^T\bxx \otimes \by_1\end{bmatrix}
		\quad\text{and}\quad
		\bh_{\varepsilon_1}^T \begin{bmatrix} \bxx' \otimes R_2^T\by_2\\R_1^T\bxx' \otimes \by_2\end{bmatrix}  = - \bh_{\varepsilon_2}^T \begin{bmatrix} \bxx' \otimes R_2^T\by_2\\R_1^T\bxx' \otimes \by_2\end{bmatrix}.
		\end{equation*}
		Since both 
		\begin{equation*}
		\begin{bmatrix} \bxx \otimes R_2^T\by_1\\R_1^T\bxx \otimes \by_1\end{bmatrix} 
		\quad \text{and}\quad
		\begin{bmatrix} \bxx' \otimes R_2^T\by_2\\R_1^T\bxx' \otimes \by_2\end{bmatrix}
		\end{equation*}
		are $\vartheta$-eigenvectors of $\LineA$, we conclude that $\bh_{\varepsilon_1}$ and $\bh_{\varepsilon_2}$ are not strongly cospectral with respect to $\LineA$.
		
		Conversely, assume that for all $\vartheta,\vartheta' \in \Lambda_{\bff_{ab}}$, $\vert \vartheta-\vartheta'\vert \neq 2$.
		It follows from Equation~(\ref{Eqn:SCh1b-2}) that, for $\lambda-2 \in \Lambda_{\bff_{ab}}$,
		\begin{equation*}
		\lambda-2 \in \Psi^{-}_{\bh_{\varepsilon_1},\bh_{\varepsilon_2}} 
		\quad \text{and}\quad
		\lambda \in \Psi^{+}_{\bh_{\varepsilon_1},\bh_{\varepsilon_2}}.
		\end{equation*}
		
		From Lemma~\ref{Lem:SCh},  $-2 \not \in \Lambda_{\bff_{ab}}$ implies $-2 \in \Psi^{+}_{\bh_{\varepsilon_1},\bh_{\varepsilon_2}}$.
		Hence $\Psi^{+}_{\bh_{\varepsilon_1},\bh_{\varepsilon_2}}  \cap  \Psi^{-}_{\bh_{\varepsilon_1},\bh_{\varepsilon_2}}  = \emptyset$
		and (\ref{Eqn:SCh1b-1}) holds.
		As a result, $\bh_{\varepsilon_1}$ and $\bh_{\varepsilon_2}$ are strongly cospectral with respect to $\LineA$.
	\end{proof}
\end{lemma}

We now consider Case~(\ref{Lem:SCh2}) of Lemma~\ref{Lem:SCh}.
\begin{lemma}
	\label{Lem:SChb}
	Let $\varepsilon_1$ and $\varepsilon_2$ be vertices of $\Line{X_1\square X_2}$ satisfying Lemma~\ref{Lem:SCh}~(\ref{Lem:SCh2}).
	Then $\bh_{\varepsilon_1}$ and $\bh_{\varepsilon_2}$ are not strongly cospectral with respect to $\LineA$.
	\begin{proof}
		Assume $\varepsilon_1$ and $\varepsilon_2$ satisfy Condition (\ref{Lem:SCh2}) in Lemma~\ref{Lem:SCh},
		and $\bh_{\varepsilon_1}$ and $\bh_{\varepsilon_2}$ are strongly cospectral with respect to $\LineA$.
		
		For $j=1,2$, let $\bxx_j$ be a $\lambda_j$-eigenvector of the signless Laplacian matrix $Q_j$ of $X_j$.
		From Equation~(\ref{Eqn:AL}),
		\begin{equation*}
		\left(\bh_{\varepsilon_1}\pm \bh_{\varepsilon_2}\right)^T  \left(R^T (\bxx_1\otimes \bxx_2) \right)= 0
		\end{equation*}
		implies 
		\begin{equation}
		\label{Eqn:SCh2C4}
		\left(\be_a^T \bxx_1\right) \otimes \left(\bff_{\alpha\beta}^T R_2^T \bxx_2\right) \pm 
		\left(\bff_{ab}^TR_1^T \bxx_1\right) \otimes \left(\be_{\alpha}^T \bxx_2\right)=0.
		\end{equation}
		Since $a$ is a pendant vertex in $X_1$, we have 
		\begin{equation*}
		\bff_{ab}^TR_1^T =\left(\be_a+\be_b\right)^T =\be_a^T Q_1.
		\end{equation*}
		Similarly, $\bff_{\alpha\beta}^TR_2^T= \be_\alpha^T Q_2$.
		Equation~(\ref{Eqn:SCh2C4}) gives 
		\begin{equation*}
		\left(\lambda_2\pm \lambda_1\right)  \left(\be_a^T \bxx_1 \otimes  \be_{\alpha}^T \bxx_2\right)=0.
		\end{equation*}
		Therefore, $\lambda_1 = \mp \lambda_2$ when $\be_a^T \bxx_1$ and $\be_{\alpha}^T \bxx_2$ are non-zero.
		Since the eigenvalues of $Q$ are non-negative, we conclude that the eigenvalue support of $\be_a$ with respect to $Q_1$
		has size at most one.   But $X_1$ is a connected graph on at least two vertices, and the eigenvalue support of $\be_a$ has
		at least two vertices.  We conclude that $\bh_{\varepsilon_1}$ and $\bh_{\varepsilon_2}$ are not strongly cospectral
	\end{proof}
\end{lemma}

We are now ready to characterize the line graphs of $X_1\square X_2$ that admit (vertex) perfect state transfer.
\begin{theorem}
	\label{Thm:VPST}
	Let $\varepsilon_1$ and $\varepsilon_2$ be edges in $X_1\square X_2$.
	Then (vertex) perfect state transfer occurs between $\bh_{\varepsilon_1}$ and $\bh_{\varepsilon_2}$ in $\Line{X_1\square X_2}$ if and only if
	the following conditions hold.
	\begin{enumerate}[(i)]
		\item
		\label{Thm:VPST1}
		$X_2=K_2$, and $X_1$ either has an cut-edge or it is a non-bipartite unicyclic graph, with
		\item
		\label{Thm:VPST2}
		$\varepsilon_1$ and $\varepsilon_2$ satisfying  Lemma~\ref{Lem:SCh}~(\ref{Lem:SCh1}), 
		\item
		\label{Thm:VPST3}
		$\Psi_{\bh_{\varepsilon_1}}  \subset \Z$ and $\Psi_{\bh_{\varepsilon_1}}\backslash\{-2\} \subset 4\Z$.
	\end{enumerate}
	\begin{proof}
		Suppose (vertex) perfect state transfer occurs between $\varepsilon_1$ and $\varepsilon_2$ in $\Line{X_1\square X_2}$.
		From Lemmas~\ref{Lem:SCh}, \ref{Lem:SCha} and \ref{Lem:SChb}, 
		Conditions~(\ref{Thm:VPST1}) and  (\ref{Thm:VPST2}) hold.
		From (\ref{Eqn:SCh1b-1}), we see that $\Psi_{\bh_{\varepsilon}}$ does not satisfy Condition~(ii) in Theorem~\ref{Thm:LinePST}~(\ref{Thm:LinePST2}).
		Hence $\Psi_{\bh_{\varepsilon_1}}  \subset \Z$.
		Let
		\begin{equation*}
		g=\gcd\{\vartheta'-\vartheta\}_{\vartheta,\vartheta' \in \Psi_{\bh_{\varepsilon_1}}}.
		\end{equation*}
		Then (\ref{Eqn:SCh1b-1}) implies  $g \vert 2$.
		By Theorem~\ref{Thm:RealPST}~(\ref{Thm:RealPST3}),  $g=2$ and $\Psi_{\bh_{\varepsilon_1}}\backslash\{-2\} \subset 4\Z$.
		
		Conversely, assume Conditions~(\ref{Thm:VPST1}) to (\ref{Thm:VPST3}) hold.   Then $\vert \vartheta-\vartheta'\vert \neq 2$, for $\vartheta,\vartheta'\in \Lambda_{\bff_{ab}}$.  
		By Lemma~\ref{Lem:SChaIFF}, $\bh_{\varepsilon_1}$ and $\bh_{\varepsilon_2}$ are strongly cospectral.
		Condition~(iii) and (\ref{Eqn:SCh1b-1}) imply that both Theorem~\ref{Thm:RealPST}~(\ref{Thm:RealPST2}) and (\ref{Thm:RealPST3}) hold.
		Hence there is perfect state transfer between $\bh_{\varepsilon_1}$ and $\bh_{\varepsilon_2}$ in $\Line{X_1\square X_2}$.
	\end{proof}
	
\end{theorem}

From Theorem~\ref{Thm:LinetoQ}, we see that the conditions given in Theorem~\ref{Thm:VPST} are sufficient for 
perfect state transfer between the plus states corresponding to edges $\varepsilon_1$ and $\varepsilon_2$ using the signless Laplacian matrix of $X_1\square X_2$ as the Hamiltonian.

\begin{remark}
	\quad
	\begin{enumerate}[(a)]
		\item
		The $n$-cube, for $n\geq 3$, gives a family of graphs in the form $X_1\square X_2$ that have perfect plus state transfer but no (vertex) perfect state transfer in its line graph.
		\item
		For $m\geq 2$, $X_1=K_{1,m}$ and $X_2=K_2$ satisfy the conditions in Lemma~\ref{Lem:SChaIFF} but not Condition~(\ref{Thm:VPST3}) of Theorem~\ref{Thm:VPST}.
		\item
		$C_4$ is the only known graph in the form $X_1\square X_2$ that has both perfect plus state transfer between antipodal edges and (vertex) perfect state
		transfer in its line graph.   
	\end{enumerate}
	A natural question is to find other graphs satisfying all conditions in Theorem~\ref{Thm:VPST}, or to show that $C_4$ is the only one.
\end{remark}
%%%%%%%%%%%%%%%%%%%%%%%%%%%%%%%%%%%%%%%%%%%%%%%%%%%%%%%%%%%%%%%%%%%%%%%%%%%%%%
%%%%%%%%%%%%%%%%%%%%%%%%%%%%%%%%%%%%%%%%%%%%%%%%%%%%%%%%%%%%%%%%%%%%%%%%%%%%%%
%%%%%%%%%%%%%%%%%%%%%%%%%%%%%%%%%%%%%%%%%%%%%%%%%%%%%%%%%%%%%%%%%%%%%%%%%%%%%%
%%%%%%%%%%%%%%%%%%%%%%%%%%%%%%%%%%%%%%%%%%%%%%%%%%%%%%%%%%%%%%%%%%%%%%%%%%%%%%
%%%%%%%%%%%%%%%%%%%%%%%%%%%%%%%%%%%%%%%%%%%%%%%%%%%%%%%%%%%%%%%%%%%%%%%%%%%%%%
%\newpage
\section{Further questions}
We list some questions arising from this paper:
\begin{enumerate}
	\item
	In Example~\ref{Ex:FR}, we rewrite Equation~(\ref{Eqn:Star}) as
	\begin{equation*}
	U_B(\tau) \left(\frac{1}{\sqrt{1+\vert s\vert^2}}\left(\be_{(a,0)}+ s \be_{(b,0)}\right)\right) =\frac{\eta}{\sqrt{1+\vert r\vert^2}}\left(\be_{(a,1)}+ r \be_{(b,1)}\right).
	\end{equation*}
	Note that $r\neq s$ if  $s \not\in \left\{2, -\frac{1}{2}\right\}$. 
	
	We propose to investigate perfect $s$-pair state transfer between $s$-pair states in the form of
	\begin{equation*}
	\frac{1}{\sqrt{1+\vert s\vert^2}}\left(\be_{a}+ s \be_{b}\right)
	\quad \text{and}\quad
	\frac{1}{\sqrt{1+\vert r\vert^2}}\left(\be_{\alpha}+ r \be_{\beta}\right),
	\end{equation*}
	where $r\neq s$.
	
	Christopher van Bommel has pointed out that in the union of $K_2$ and $C_4$, 
	for any vertex $a$ in $C_4$ and $b$ in $K_2$, there is adjacency perfect $s$-pair state transfer from $\be_a+\be_b$ to $\be_a-\be_b$ at time $\pi$.
	Does there exist a connected graph admitting adjacency perfect $s$-pair state transfer from $\be_a+\be_b$ to $\be_\alpha-\be_\beta$?
	
	\item
	Example~\ref{Eg:P2ST}~(\ref{Eg:Pal}) gives a family of trees admitting adjacency pair state transfer.
	Is it possible to have adjacency perfect $s$-pair state transfer in trees from initial state $\be_a + s \be_b$ where $s\neq -1$?
	Is it possible to have Laplacian perfect $s$-pair state transfer in trees?
	
	\item
	In Section~\ref{Subsection:DRG}, we determine all instances of perfect $s$-pair state transfer in distance regular graphs that admit (vertex) perfect state transfer.
	The cycle $C_8$ is an example of a distance-regular graph admitting perfect $s$-pair state transfer but it has no (vertex) perfect state transfer nor fractional revival.
	We ask for the characterization of distance-regular graphs that have perfect $s$-pair state transfer.  In particular, determine if perfect $s$-pair state transfer
	can occur in a primitive distance-regular graph.
	%\cite{book:GraphAndDigraphs}
\end{enumerate}

%%%%%%%%%%%%%%%%%%%%%%%%%%%%%%%%%%%%%%%%%%%%%%%%%%%%%%%%%%%%%%%%%%%%%%%%%%%%%%
%%%%%%%%%%%%%%%%%%%%%%%%%%%%%%%%%%%%%%%%%%%%%%%%%%%%%%%%%%%%%%%%%%%%%%%%%%%%%%
%%%%%%%%%%%%%%%%%%%%%%%%%%%%%%%%%%%%%%%%%%%%%%%%%%%%%%%%%%%%%%%%%%%%%%%%%%%%%%
%%%%%%%%%%%%%%%%%%%%%%%%%%%%%%%%%%%%%%%%%%%%%%%%%%%%%%%%%%%%%%%%%%%%%%%%%%%%%%
%%%%%%%%%%%%%%%%%%%%%%%%%%%%%%%%%%%%%%%%%%%%%%%%%%%%%%%%%%%%%%%%%%%%%%%%%%%%%%
%\newpage
\noindent
\emph{Acknowledgement}: We thank David Feder, Chris Godsil, Alastair Kay, Christopher van Bommel and Christino Tamon for useful discussions. A. Chan gratefully acknowledges the support of the NSERC Grant No. RGPIN-2021-03609. S. Kim is supported in part by the Fields Institute for Research in Mathematical Sciences and NSERC. S. Kirkland is supported by NSERC grant number RGPIN–2019–05408. H. Monterde is supported by the University of Manitoba Faculty of Science and Faculty of Graduate Studies. S. Plosker is supported by NSERC Discovery grant number RGPIN-2019-05276, the Canada Research Chairs Program grant number 101062, and the  Canada Foundation for Innovation grant number 43948.

%\bibliographystyle{alpha}
%\bibliography{QIPSubmission}

\begin{thebibliography}{AGKM06}
	
	\bibitem[AGKM06]{Akbari}
	Saieed Akbari, Narges Ghareghani, Gholamreza~B. Khosrovshahi, and Hamidreza
	Maimani.
	\newblock The kernels of the incidence matrices of graphs revisited.
	\newblock {\em Linear Algebra Appl.}, 414(2-3):617--625, 2006.
	
	\bibitem[BCN89]{BCN1989}
	A.~E. Brouwer, A.~M. Cohen, and A.~Neumaier.
	\newblock {\em Distance-regular graphs}, volume~18 of {\em Ergebnisse der
		Mathematik und ihrer Grenzgebiete (3) [Results in Mathematics and Related
		Areas (3)]}.
	\newblock Springer-Verlag, Berlin, 1989.
	
	\bibitem[Bos03]{Bose:Quantum}
	Sougato Bose.
	\newblock {Quantum communication through an unmodulated spin chain}.
	\newblock {\em Physical Review Letters}, 91(20):207901, 2003.
	
	\bibitem[CCT{\etalchar{+}}19]{Chan2019}
	Ada Chan, Gabriel Coutinho, Christino Tamon, Luc Vinet, and Hanmeng Zhan.
	\newblock {Quantum fractional revival on graphs}.
	\newblock {\em Discrete Applied Mathematics}, 269:86--98, 2019.
	
	\bibitem[CDEL04]{Christandl:PSTonHypercubes}
	Matthias Christandl, Nilanjana Datta, Artur Ekert, and Andrew~J. Landahl.
	\newblock {Perfect state transfer in quantum spin networks}.
	\newblock {\em Physical Review Letters}, 92(18):187902, 2004.
	
	\bibitem[CG20]{Chen2020PairST}
	Qiuting Chen and Chris Godsil.
	\newblock Pair state transfer.
	\newblock {\em Quantum Information Processing}, 19:321, 2020.
	
	\bibitem[CG21]{Coutinho2021}
	Gabriel Coutinho and Chris Godsil.
	\newblock {Graph Spectra and Continuous Quantum Walks}.
	\newblock 2021.
	
	\bibitem[CGGV15]{Coutinho2015}
	Gabriel Coutinho, Chris Godsil, Krystal Guo, and Frederic Vanhove.
	\newblock {Perfect state transfer on distance-regular graphs and association
		schemes}.
	\newblock {\em Linear Algebra and Its Applications}, 478:108--130, 2015.
	
	\bibitem[Che19]{chen2019edge}
	Qiuting Chen.
	\newblock Edge {S}tate {T}ransfer.
	\newblock Master's thesis, University of Waterloo, 2019.
	
	\bibitem[Cou14]{Coutinho2014}
	Gabriel Coutinho.
	\newblock {Quantum State Transfer in Graphs}.
	\newblock {\em Ph.D. Dissertation}, 2014.
	
	\bibitem[FTL22]{fan2022stabilizing}
	Yi-Zheng Fan, Meng-Yu Tian, and Min Li.
	\newblock The stabilizing index and cyclic index of the coalescence and
	cartesian product of uniform hypergraphs.
	\newblock {\em Journal of Combinatorial Theory, Series A}, 185:105537, 2022.
	
	\bibitem[God12]{Godsil2010}
	Chris Godsil.
	\newblock {When can perfect state transfer occur?}
	\newblock {\em Electronic Journal of Linear Algebra}, 23:877--890, 2012.
	
	\bibitem[God17]{Godsil2017RST}
	Chris Godsil.
	\newblock {Real State Transfer}.
	\newblock {\em arXiv:1710.04042}, 2017.
	
	\bibitem[GS24]{GodsilSmith2024}
	Chris Godsil and Jamie Smith.
	\newblock Strongly cospectral vertices.
	\newblock {\em Australas. J. Combin.}, 88:1--21, 2024.
	
	\bibitem[GZ22]{Zhang2024}
	Chris Godsil and Xiaohong Zhang.
	\newblock Fractional revival on non-cospectral vertices.
	\newblock {\em Linear Algebra Appl.}, 654:69--88, 2022.
	
	\bibitem[KMP23]{Monterde2022}
	Steve Kirkland, Hermie Monterde, and Sarah Plosker.
	\newblock Quantum state transfer between twins in weighted graphs.
	\newblock {\em Journal of Algebraic Combinatorics}, 58(3):623--649, 2023.
	
	\bibitem[Leh33]{lehmer1933trig}
	D.H. Lehmer.
	\newblock A note on trigonmetric algebraic numbers.
	\newblock {\em The American Mathematical Monthly}, 40(3):165--166, 1933.
	
	\bibitem[Pal24]{Pal2024}
	Hiranmoy Pal.
	\newblock {Quantum Pair State Transfer on Isomorphic Branches}.
	\newblock {\em arXiv:2402.07078}, 2024.
	
	\bibitem[Shi07]{shi2007Lbound}
	Lingsheng Shi.
	\newblock Bounds on the ({L}aplacian) spectral radius of graphs.
	\newblock {\em Linear Algebra Appl.}, 422(2-3):755--770, 2007.
	
\end{thebibliography}

\end{document}